\pdfoutput=1
\documentclass[10pt]{article}
\usepackage{latexsym, amscd, amsfonts, eucal, mathrsfs, amsmath, amssymb, amsthm, xr,  makeidx, stmaryrd}
\usepackage[letterpaper, portrait, margin=1in]{geometry}

\usepackage[usenames]{color}
\usepackage{verbatim}
\usepackage{tensor}
\usepackage{enumitem}
\usepackage[hidelinks]{hyperref}   
\usepackage{epigraph}
\usepackage{graphicx}
\usepackage{tikz}

\usepackage{verbatim}
\usepackage{tensor}
\usepackage{enumitem}
\usepackage{subcaption}
\usetikzlibrary{backgrounds}
\usepackage{float}
\newtheorem{prop}{Proposition}
\newtheorem{theorem}{Theorem}

\newtheorem{lemma}{Lemma}

\usepackage[style=philosophy-classic,backend=biber]{biblatex}
\def\E{\mathcal{E}}
\def\D{\mathcal{D}}
\usepackage{turnstile}
\usepackage{cancel}
\usepackage{amssymb}
\usepackage{scalerel}
\DeclareMathOperator*{\lsum}{\scalerel*{\oplus}{\textstyle\sum}}

\usepackage{mathtools}
\usepackage{stackengine}
\setstackEOL{\\}
\usepackage{tikz}
\usepackage{tikz-3dplot}
\usepackage{tikz-cd}
\usetikzlibrary{arrows,decorations.pathmorphing,backgrounds,positioning,fit,petri,calc,shapes.misc,decorations.markings,intersections}
\tikzset{degil/.style={
                decoration={markings,
                mark= at position 0.5 with {
                \node[transform shape] (tempnode) {$\backslash$};
                }
                },
                postaction={decorate}
}
}

\usepackage{enumitem}
\usepackage{caption}
\usepackage{array}

\numberwithin{equation}{section} % Number equations within sections (i.e. 1.1, 1.2, 2.1, 2.2 instead of 1, 2, 3, 4)
\numberwithin{figure}{section} % Number figures within sections (i.e. 1.1, 1.2, 2.1, 2.2 instead of 1, 2, 3, 4)
\numberwithin{table}{section} % Number tables within sections (i.e. 1.1, 1.2, 2.1, 2.2 instead of 1, 2, 3, 4)

\def\D{\mathcal{D}}

\def\M{\mathcal{M}}

\def\R{\mathbb{R}}
\def\V{\mathcal{V}}

\def\E{\mathcal{E}}

\def\p{\mathfrak{p}}

\def\<{\langle}
\def\>{\rangle}

\def\x{\mathbf{x}}

\def\Ebar{\bar{E}}

\definecolor{red}{rgb}{0.9,0,0}
\definecolor{green}{rgb}{0,0.8,0}
\definecolor{blue}{rgb}{0,0,0.8}
\definecolor{cautionred}{rgb}{1.0,0,0}

\newcommand{\Sum}{\sum}

\newcommand{\bea}{\begin{eqnarray}}
\newcommand{\eea}{\end{eqnarray}}
\newcommand{\beann}{\begin{eqnarray*}}
\newcommand{\eeann}{\end{eqnarray*}}
\newcommand{\beq}{\begin{equation}}
\newcommand{\eeq}{\end{equation}}
\newcommand{\ba}{\begin{array}}
\newcommand{\ea}{\end{array}}
\newcommand{\ben}{\begin{enumerate}}
\newcommand{\een}{\end{enumerate}}

\theoremstyle{definition}

\theoremstyle{definition}

\theoremstyle{definition}

\theoremstyle{definition}
\theoremstyle{definition}
\theoremstyle{definition}

\theoremstyle{remark}

\theoremstyle{definition}

\begin{document}

\title{Two Forms of Inconsistency in Quantum Foundations}
 \author{Jeremy Steeger\\
  \texttt{jsteeger@nd.edu}
  \and
  Nicholas Teh\\
  \texttt{nteh@nd.edu}
  }
\maketitle
\begin{abstract}
   
   Recently, there has been some discussion of how Dutch Book arguments might be used to demonstrate the rational incoherence of certain hidden variable models of quantum theory (\cite{Feintzeig2014, Feintzeig2017, Wronski2017}).
   In this paper, we argue that the `form of inconsistency' underlying this alleged irrationality is deeply and comprehensively related to the more familiar `inconsistency' phenomenon of contextuality. Our main result is that the hierarchy of contextuality due to Abramsky and Brandenburger (\citeyear{AB}) corresponds to a hierarchy of additivity/convexity-violations which yields formal Dutch Books of different strengths. 
   We then use this result to provide a partial assessment of whether these formal Dutch Books can be given a convincing normative interpretation. 
    %There is an emerging literature applying Dutch book arguments to quantum theories. It seems reasonable to suspect there might be a tight link between this measure of incoherence and the manner in which theories where states fix measurement-outcomes are usually thought to be inconsistent---namely, such outcomes cannot remain independent of the context of measurement. We establish such a link, demonstrating how a hierarchy of contextuality due to \cite{AB} can be recovered as a hierarchy of conditions yielding Dutch Books in the context of weak probability spaces (\cite{Feintzeig2017}). We then assess the manner in which these Dutch book conditions do---and do not---serve as markers of the rationality of the theories in question.
\end{abstract}

\section{Introduction}

Within the philosophical foundations of quantum theory, there is a long tradition of attempting to argue that the predictions of quantum theory cannot be reproduced by hidden variable theories that assume a classical structure of events, probabilities of events, states of affairs and functional relationships relating these items (call these \textbf{CHVTs}). 
% \footnote{A note to specialists: Here, `states of affairs' can be understood as the `ontic states' in a Spekkens-style ontological model. See Section 5 for the details and a discussion of our main results with respect to the ontological model framework.}
Very roughly speaking, such `no-go' arguments take the following form:
\begin{quote}
\textbf{(Standard)} There exists a quantum model whose empirical probabilities are inconsistent with the probabilities that can be reproduced by the relevant CHVT.\footnote{We will give more detail in Sections 2 and 3 about the class of CHVTs that we wish to consider. In brief: we will only consider \textit{noncontextual} `factorizable CHVTs' in the sense of (\cite[Section 8]{AB}); thus, (Standard) arguments do not apply to Bohmian mechanics, which is both contextual and non-local.}
\end{quote}
For example, the GHZ, Kochen-Specker, Hardy, and Bell quantum models all provide instances of (Standard) arguments.

The development of this tradition has since resulted in a veritable cornucopia of models witnessing such inconsistencies with classical intuitions.
Thus, one important avenue of contemporary research (\cite{AB, Spekkens2011, Sainz2012}) seeks to impose order on this landscape by uncovering what we shall refer to as `forms of inconsistency', i.e. high-level structures that provide a unifying explanation---and perhaps even a classification---of how these models give rise to such inconsistencies.
The goal of this paper is to argue for a thoroughgoing and novel connection between two forms of inconsistency in quantum foundations, and to apply this result to several philosophical puzzles.
%As we will see, the high-level structure disclosed by such investigations is intimately related to how one conceptualizes the `inconsistencies' that underlie no-go arguments; henceforth, we will use the term `form' to refer to a particular conceptualization of this sort.
%The goal of this paper is to argue for a thoroughgoing and novel connection between two forms of inconsistency in quantum foundations, and to apply this result to resolve several philosophical puzzles.

The first form, which has a long and distinguished pedigree, is the \textit{Gluing Inconsistency} that is suggested by `contextuality arguments' against CHVTs, as exemplified by the Kochen-Specker theorem (\cite{Kochen1975}). 
Its structural essence has recently been rigorously articulated by Abramsky and Brandenburger (henceforth \textbf{AB}), who not only generalize it to apply to a large class of physical theories, but also provide a three-tiered `hierarchy of contextuality' which classifies the ways in which physical models can deviate from the classical probability assumed by a CHVT.

%and it is AB's precisification of Gluing Inconsistency that we will appeal to throughout our paper. 
The second form, which we call \textit{Dutch Bookability}, is somewhat more recent and thus less explored; indeed, much of the work of this paper will be to show that it really is a `form of inconsistency' in our sense.
In its simplest incarnation, it arises in (\cite{Feintzeig2014}) and (\cite{Feintzeig2017}) and appears to differ from the first form in two ways. First, it is not presented as arising from an argument against a CHVT, but rather against a \textit{Generalized HVT} \textbf{(GHVT)}, by which we mean a hidden variable theory based on a weakening of classical probability theory, in a sense that we will make precise below.
Second, the `inconsistency' in question stems from the fact that the GHVT is `Dutch Bookable' in the sense that it violates a certain formal constraint. Feintzeig and Fletcher (\citeyear{Feintzeig2017}) then make a move that has become standard in philosophical discussions concerning \textit{classical} probability: they interpret this constraint normatively, i.e. as a rational constraint on the expectations of an epistemic agent, and thus take its violation to provide a no-go argument against a subjective interpretation of the probabilities of GHVTs.\footnote{Feintzeig and Fletcher also argue that violating the formal constraint associated with `Dutch Books' is pathological even when the probabilities in question are given a frequentist interpretation, but we will not address this argument here.}
For the sake of clarity, we will distinguish between the formal constraint and its normative intepretation by using \textit{`Dutch Bookability' to refer to the violation of the formal constraint}, and `normative Dutch Bookability' to refer to the normative interpretation of this violation.

Despite these apparent differences, we will argue that Gluing Inconsistency and Dutch Bookability are really two ways of understanding the same set of facts which fall under (Standard), because each form allows us to represent the empirical probabilities of a physical model and to classify the ways in which certain HVTs fail to reproduce these probabilities. 
% \footnote{More carefully, we take this to be the minimal interpretation of Dutch Bookability; nothing yet that we have said militates against its further interpretation as providing a no-go argument against GHVTs.}
Our argument culminates in the result (Theorem \ref{thm:main}) that, for a large class of physical models (including quantum models), AB's three-tiered hierarchy of Gluing Inconsistency corresponds to a three-tiered hierarchy of Dutch Bookability, which provides the aforementioned classification.

This result yields powerful applications to two puzzles raised by the recent exploration of the relationship between Dutch Books and quantum theory.
First, the only connection between GHVT models of quantum theory and Dutch Bookability that has appeared in the literature (\cite{Feintzeig2014,Feintzeig2017}) stems from a `finite null cover'---an incredibly strong (indeed maximal!) violation of `subadditivity' which leads to Dutch Books. 
% \footnote{We note that in the standard Dutch Book literature, it is typical to discuss the violation of `additivity' (both finite and countable), but we have not been able to find discussions of subadditivity-violation, let alone \textit{maximal} violations of this kind!}
What explanation can be given of the strength of this violation?
Our main result allows us to answer this question by explaining the strength of the violation in terms of AB's three-tiered classification of the contextuality exhibited by quantum (and non-quantum) models: the `maximal Dutch Bookability' of the above violation derives from its being the strongest form of contextuality (exemplified in quantum theory by Kochen-Specker models), but there are also quantum models which exhibit weaker forms of Dutch Bookability corresponding to the second and third tier respectively, viz. the Hardy model (\cite{Hardy1993}) and the Bell model (\cite{Bell1964}). 

Second, what room is there to interpret Dutch Bookability normatively, i.e. as a no-go argument against a subjective approach to the probabilities of noncontextual GHVTs?
This question is difficult to fully assess because the relevant hidden variable theories have not been clearly defined in the literature. Nonetheless, we will suggest the normative force of Dutch Bookability hinges on an approach to GHVTs that borrows an interpretive strategy from classical probability theory---and this approach is at the very least controversial.

%and more familiar (at least to philosophers!) because the purported inconsistency is one that has been widely 

The paper is structured as follows.
First, Section 2 provides a non-technical introduction to our topic that is accessible to the general reader. Section 2.1 reviews the notion of `contextuality' that is central to this paper. Section 2.2 introduces the first form of inconsistency, viz. AB's Gluing Inconsistency, and Section 2.3 introduces the second form of inconsistency, viz. (formal) Dutch Bookability. Section 2.3 then gives a conceptual explanation of the main results of our paper and their motivation. At this point, the general reader will have sufficient background to skip straight ahead to the conclusion (Section 5).

Sections 3 and 4 are aimed more at specialist readers in the philosophy of physics and formal epistemology. However, we stress that, even here, technical prerequisites have been kept to an absolute minimum and our arguments are accessible to a reader with only knowledge of elementary set theory. Thus, a general reader who is willing to work through the details should be able to follow the gist of our arguments. 

Section 3.1 provides an elementary review of AB's highly influential framework for contextuality, which demonstrates that contextuality comes in different strengths or `tiers' of a hierarchy. To the best of our knowledge, this is the first such treatment aimed at philosophers. 
Section 3.2 then introduces one of our novel contributions, viz. the construction of a translation manual from AB's framework to the `weak probability spaces' that are used in order to formulate Dutch Bookability. The details of this translation manual are subtle and themselves of conceptual significance, since they describe how to give a quasi-probabilistic interpretation to one of the deepest features of quantum theory, viz. contextuality.

Section 4 draws on the material of Section 3 to prove our main result, viz. the correspondence between AB's hierarchy for contextuality and a hierarchy of (formal) Dutch Bookability. 
We begin by situating our argument with respect to previous work on a hierarchy of Dutch Bookability that was conducted within a purely asbtract setting, i.e. that did not relate this hierarchy to physical models (and quantum theory in particular). 
By contrast, Section 4.1 relates the AB hierarchy---and its connection with \textit{physical models}---to a hierarchy of Dutch Bookability in terms of `additivity-violation', and Section 4.2 does the same for a hierarchy of Dutch Bookability in terms of `convexity-violation'. We note that the formal epistemology community will be especially interested in the results of Section 4.2, since the standard modern approach to Dutch Books is by means of convex analysis.\footnote{This approach goes back at least to de Finetti (\citeyear{deFinetti-F:1937,deFinetti-ToP:1974}) and can be seen more recently in Paris (\citeyear{Paris2001}), K\"{u}hr and Mundici (\citeyear{Kuhr2007}), Williams (\citeyear{Williams2012}), and Bradley (\citeyear{Bradley2016}).}

% Furthermore, it is very likely that our techniques are applicable to other areas of formal epistemology, such as the recently explored subject of `fragmented decision theory' CITE!!.

Finally, Section 5 addresses the question of how formal Dutch Bookability may be interpreted normatively, i.e. as a no-go against a subjective interpretation of the probabilities of GHVTs. We first offer one plausible reconstruction of Feintzeig and Fletcher's normative Dutch Book argument, and we show that the soundness of this reconstructed argument hinges on a certain choice of response functions, i.e. probabilities of events conditional on certain states of affairs. We then note that a GHVT need not specify response functions in this way.

\section{Contextuality and Inconsistency}\label{sec:con_in}

\subsection{Contextuality}\label{subsec:context}

The subject of `contextuality' was introduced into quantum foundations in (\cite{Bell1966}) and (\cite{Kochen1975}), after which it has become common to characterize quantum theory as `contextual'.
In order to explain what this means, we now review several preliminary definitions.
First, a (measurement) \textit{context} is a set of `co-measurable' or `compatible' measurements, i.e. measurements that can be made jointly (thus, in the particular case of quantum theory, a `context' refers to a set of commuting observables).
Second, a \textit{noncontextual} \textbf{(NC)} CHVT is one whose `response functions' (i.e. functions that specify the probabilities of measurement events conditional on the system's being in some particular state) are independent of any information about contexts.\footnote{We note that Spekkens (\citeyear{Spekkens2005}) generalizes the above definitions of `context' and `noncontextual CHVT' to apply to contexts of preparations and transformations of systems as well; however, we will only be considering measurement contexts in this paper.}
Third, we define a \textit{possibility assignment} as an assignment of either $1$ or $0$ to events. Following Spekkens (\citeyear{Spekkens2005}), we will use the term \textit{outcome-deterministic} \textbf{(OD)} to refer to an CHVT whose response functions only make possibility assignments.

The characterization of quantum theory as `contextual' then arises through the following `quantum contextuality argument', which is a special case of (Standard): 
there exists a quantum model whose empirical possibilities are inconsistent with the response functions of an (NC, OD) CHVT.
We say that quantum theory is `contextual' because such quantum models exist.

We now discuss how quantum contextuality suggests a geometric form of inconsistency that we shall call `Gluing Inconsistency'. 
One classic way of demonstrating quantum contextuality is by means of a (discrete) Kochen-Specker model, i.e. a quantum model for which `measurements' correspond to (a finite set of) projectors, `maximal contexts' are given by sets of mutually orthogonal projectors that resolve to the identity, and the empirical possibilities are assumed to satisfy certain functional constraints.\footnote{In particular, the assignment $h: \mathcal{P(H)}\to\{0,1\}$ of possibilities must be normative ($h(I)=1$) and additive ($h\left (\bigvee_i P_i\right)=\sum_i h(P_i)$ for $\{P_i\}$ mutually orthogonal).}
The inconsistency between the empirical probabilities of a Kochen-Specker model and those generated by a (NC) possibilistic CHVT can then be rendered geometrically, i.e. as the impossibility of consistently `gluing together' orthonormal bases (corresponding to the maximal contexts) with the following possibility assignments: each orthonormal basis has exactly one basis vector that is assigned the possibility $1$ and all other basis vectors are assigned the possibility $0$. 
For instance, an impossible gluing scenario that corresponds to the Kochen-Specker model in (\cite{Cabello1996}) is depicted in Fig. \ref{fig:Cabello}, where the vertices represent possible measurement events, the rectangles represent maximal contexts, and the vertices are colored black for possibility $1$, and white for possibility $0$.

\begin{figure}[!htb]
    \centering
    \begin{tikzpicture}[scale=0.50]
        \newdimen\R
        \newdimen\nodesize
        \R=2.7 cm
        \nodesize=2 pt
        
        % corners
        \foreach \x/\l/\p in
         {60/a/above,120/b/above, 180/c/left,240/d/below,300/e/below,360/f/right}
         \coordinate [label=\p: ] (\l) at (\x:\R);
        
        % sides
        \foreach \x/\l/\p in {a/b/ab,b/c/bc, c/d/cd,d/e/de,e/f/ef,f/a/fa}
        \path[name path global/.expanded=\p] (\x)--(\l);
        
        % intersection paths
        \foreach \x/\l in
         {80/abde1,100/abde2, 140/bcef1,160/bcef2,200/cdfa1,220/cdfa2}
         \path[name path global/.expanded=\l] (\x:-\R)--(\x:\R);
        
        % intersections
        \foreach \x/\l/\p in
         {ab/abde1/ab1,ab/abde2/ab2,
         de/abde1/de1,de/abde2/de2,
         bc/bcef1/bc1,bc/bcef2/bc2,
         ef/bcef1/ef1,ef/bcef2/ef2,
         cd/cdfa1/cd1,cd/cdfa2/cd2,
         fa/cdfa1/fa1,fa/cdfa2/fa2} \path [name intersections/.expanded = {of = {\x} and \l, by = \p}];
        
        % nodes
        \foreach \x in
         {a, b, c, d, e, f, ab1, ab2, bc1, bc2, cd1, cd2, de1, de2, ef1, ef2, fa1, fa2}
         \draw (\x) circle (\nodesize);
         
         % contexts
         \node[draw,fit=(a) (b)] {};
         \node[draw,rotate fit=-30, fit=(b) (c)] {};
         \node[draw,rotate fit=30, fit=(c) (d)] {};
         \node[draw,rotate fit=0, fit=(d) (e)] {};
         \node[draw,rotate fit=-30, fit=(e) (f)] {};
         \node[draw,rotate fit=30, fit=(f) (a)] {};
         \node[draw,rotate fit=90, fit=(bc2) (cd1) (ef2) (fa1)] {};
         \node[draw,rotate fit=-30, fit=(ab1) (fa2) (cd2) (de1)] {};
         \node[thick,draw,rotate fit=30, fit=(ab2) (bc1) (de2) (ef1)] {};
         
         % coloring
         \draw[black,fill=black] (e) circle (\nodesize);
         \draw[black,fill=black] (fa2) circle (\nodesize);
         \draw[black,fill=black] (b) circle (\nodesize);
         \draw[black,fill=black] (cd1) circle (\nodesize);
        
    \end{tikzpicture}
    \caption{An inconsistent possibility assignment for Cabello's Kochen-Specker model.}%Cabello's 18-vector Kochen-Specker model. Each node represents a unit vector in a four-dimensional Hilbert space (these can be taken to represent various outcomes for spin-measurements of entangled electrons); each box represents a (maximal) context, i.e. an orthonormal basis comprised of four such vectors. For each context, an (NC, OD) HVT needs to assign `black' (1) to exactly one node, and `white' (0) to the other nodes. It is easy to see that no global assignment can consistently satisfy this constraint.}
    \label{fig:Cabello}
\end{figure}
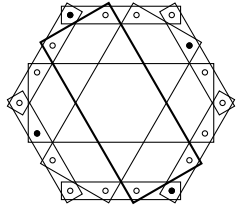

Such models present us with a compelling image of inconsistency which admits of an intuitive parallel with the global geometric inconsistency of `impossible figures' such as Penrose's Triangle (\cite{Penrose1992}) and Escher's staircase.
As we will see in Section 3.1, AB's account of Gluing Inconsistency provides both a precisification and a generalization of this core geometric intuition.

%The idea of inconsistency as a `failure of gluing' is dramatically visualizable, and admits of an intuitive parallel with the global `inconsistency' of impossible structures such as 

%One virtue of theory-independence is that it paves the way for a more refined framing of standard arguments against HVTs in
%One trend that has emerged in contemporary research on contextuality is the attempt to frame contextuality arguments in a more refined manner:
%\begin{quote} (Schema) 
%Codify the empirical probabilities in terms of a mathematical structure $M$ that allows one to %define a property $K$ whose presence serves as an \textit{obstruction} to the existence of a %certain HVT that reproduces the empirical probabilities.\footnote{Briefly discuss obstruction to %X as a property that entails the negation of a necessary condition for X.}
%\end{quote}
%The property $K$ embodies the \textit{form} of the inconsistency between the empirical probabilities and the HVT.
%Why introduce such conceptual refinements? Gives us cognitive control; taxonomy. Virtues need not be the same. 

%We now turn to a discussion of one such form of inconsistency that has a particularly long pedigree within the history of contextuality, and which has recently been given a general and precise formulation in ****. 
%Mention here Belen's et al's graph-theoretic approach. 

\subsection{Gluing Inconsistency}

Although the geometric intuition underlying Gluing Inconsistency has long been implicit in the quantum contextuality tradition, it was not until relatively recently that Abramsky and Brandenburger \textbf{(AB)} used a rigorous framework---viz. sheaf theory---to distill its essence and generalize it far beyond the confines of quantum contextuality.  

The generalization offered by AB's framework is one that we shall refer to as \textit{theory-independence}, because it allows us to consider models of a large class of theories, including \textit{non-quantum} theories.
Theory-independence can be motivated by noting that (Standard) arguments function by demonstrating an inconsistency between the `structure of the empirical probabilities'---which happen to be generated by a quantum model---on the one hand, and a certain noncontextual CHVT, on the other.
This observation suggests that one might be able to devise theory-independent versions of such arguments, in which the empirical probabilities do not necessarily come from quantum theory.

In order to do so, AB introduce a formal generalization of `the empirical probabilities of a quantum model', which they call an `empirical model'.
Roughly speaking, an empirical model is a family of (classical) probability distributions that satisfies a generalized no-signaling condition (cf. Section 3.1 for the details).
\textit{Gluing Inconsistency} is then defined as the failure of such a family of distributions to `glue together' to form a global/joint distribution. 
Furthermore, since AB show that the existence of such joint distributions is equivalent to the existence of a certain CHVT, this `gluing failure' can be understood as a measure of the extent to which an empirical model deviates from classical probability.\footnote{As AB note, their work thus extends arguments for quantum models along these lines that date back at least to Fine (\citeyear{Fine1982a}).}

One of the most attractive features of AB's framework is that it clearly picks out certain features of empirical models which characterize the degree to which the model exhibits Gluing Inconsistency; henceforth, we will use the term `contextual' to describe models with such features (we will continue to use the term `quantum contextuality' to describe when quantum models deviate from classical probability).
AB (\citeyear{AB}) were thus able to provide the first comprehensive classification of contextuality in terms of a three-tiered hierarchy, where a higher tier (a stronger failure of gluing) implies a lower tier (a weaker failure of gluing), but not vice versa.
As we will soon see, the top and middle tier of the hierarchy involve only empirical \textit{possibilities}, whereas the lowest tier of this hierarchy is fully \textit{probabilistic} in its characterization.
In particular, the GHZ model and all Kochen-Specker models fall into the top tier, the Hardy model falls into the middle tier (but not the top), and the Bell model falls into the lowest tier (but not the middle).\footnote{For the reader familiar with (\cite{Feintzeig2014}), we briefly note the heuristic relationship of this hierarchy with Feintzeig's hierarchy of theorems `Kochen-Specker $\Rightarrow$ Theorem 6 $\Rightarrow$ Pitowsky's Theorem $\Rightarrow$ Bell's Theorem'. The statement that `Kochen-Specker $\Rightarrow$ Theorem 6' concerns top-tier contextuality and additivity/convexity violations, and the statement that `Pitowsky's Theorem $\Rightarrow$ Bell's Theorem' concerns lowest-tier contextuality and additivity/convexity violations (we will detail these tiers, as well as the middle tier, below).}
In sum, the `AB hierarchy' encompasses all known quantum models that witness (Standard), and classifies the different strengths with which they deviate from classical probability.

%In fact, the generality afforded by AB's framework for Gluing Inconsistency leads to a three-tiered hierarchy of ways in which empirical probabilities can deviate from classicality.
%The top two tiers of this hierarchy involve only empirical \textit{possibilities}, and thus encompass the phenomenon of contextuality as we have defined it.
%The top tier is called Strong Contextuality, and it includes standard examples of contextuality such as Kochen-Specker models and the GHZ model, as well as  \textit{non-quantum} models such as the PR Box.
%The middle tier is called Logical Contextuality, and it includes examples such as Hardy's model. 
%Finally, the lowest tier, i.e. the weakest deviation from classicality, is fully \textit{probabilistic} in nature, and it includes Bell's classic model of non-locality. 

%We have have surveyed three features in virtue of which we take AB's Gluing Inconsistency to count as a `form of inconsistency': (i) It is theory-independent; (ii) it is mathematically precise; and (iii) it provides an illuminating classification of different

%The ability to articulate such a hierarchy is one of the most powerful features of AB's approach %to Gluing Inconsistency, and it contribute one of the main ingredients of our result

\subsection{Dutch Bookability}

While the notion of Gluing Inconsistency has a straightforward relationship with the quantum contextuality tradition, the same cannot be said for our second form of inconsistency, viz. Dutch Bookability.
But a relationship is present nonetheless, and is most clearly seen by dividing an account of Dutch Bookability into two parts, viz. (i) the partial algebra approach; and (ii) the Dutch connection.
%straightforwardly arises out of a geometric intuition that has long been implicit in theorizing about contextuality, the relationship between Dutch Bookability and Tradition is somewhat more complex.
\\\\
\emph{(i) The partial algebra approach} 

The key formal notion that Feintzeig and Fletcher (\citeyear{Feintzeig2017}) use to establish a link between Dutch Books and quantum foundations is that of a Weak Probability Space (WPS), i.e. a generalization of a measure space to a setting in which the `set of events' $\Sigma$ is not required to have any algebraic structure (see Section 3.2 for the precise definition of a WPS) and the `probability function' $\mu: \Sigma \rightarrow [0,1]$ is a mere set function and not a measure. 
However, their actual use of a WPS to represent quantum models still requires \textit{some} algebraic structure, because \textit{co-measurable} WPS events (i.e. those that arise from measurements within the same context) are required to form an algebra. 
Thus, these WPS representations possess a \textit{partial} algebraic structure in which certain algebraic relations model the co-measurability of events stemming from compatible measurements, and the absence of these algebraic relations models the non-co-measurability of events stemming from incompatible measurements---we shall call this the `partial algebra approach'.
We stress that the WPSs with partial algebraic structure need not have anything to do with generalizations of hidden variable theories! In fact, we submit that in order to achieve complete conceptual clarity on the relation between (formal) Dutch Bookability and contextuality, it is best to simply think of a WPS as a formal device for \textit{representing} quantum theory, akin to AB's use of `empirical models'. Furthermore, only after getting clear on this relation will one be in a position to assess whether it makes sense to interpret WPSs as a kind of hidden variable theory.

The partial algebra approach has a long history within quantum foundations that is not motivated by Dutch Books, but rather by the recognition that (because of non-commuting observables) quantum models will in general give rise to measurements that fall into distinct maximal contexts; in other words, partial algebras are meant to capture the very structure that gives rise to `quantum contextuality'.
This use of partial algebras goes at least as far back as the original Kochen-Specker paper (\citeyear{Kochen1975}), which employs the notion of a partial Boolean algebra for this purpose.
Since then, the partial algebra approach (and its associated measure theory) has been developed in various ways.
One strand (\cite{Staton2015, Heunen2010, Roumen2016}) focuses on developing the approach in a purely algebraic manner, wherein events do not have to be realized by sets. This work has resulted in elegant (indeed, category-theoretic) constructions and powerful characterization results. 

A different strand attempts to cleave very closely to the standard notion of a measure space, in which events are realized by sets.
As such, it has to reckon with minutiae such as first specifying an ambient set $Y$ that serves as the `sample space' in order to specify the set of events (some subset of the power-set of $Y$) that carries the partial structure. One attempt to work this out in detail is provided by the notion of a `generalized measure space' in (\cite{Gudder1973}). 
Since a WPS is a weakening of this notion, it is helpful to situate it within this strand of the partial algebra approach.
%which thus has to account for hairy notion of a WPS stems from a development of this idea that cleaves more closely to traditional set-theoretic presentations of probability theory
%For instance, one of the predecessors of this notion is Gudder's GPS....
%WPS as a generalization of Gudder's framework.
\\\\
\emph{(ii) The Dutch connection}

The connection between Dutch Books and the partial algebra approach has an antecedent in the work of Fine (\citeyear{Fine1982a}), who can be read as suggesting that one might be able to avoid (Standard) arguments against CHVTs if one generalizes these CHVTs to account for non-co-measurability by appealing to the partial algebraic framework for probability. We will use the term \textit{Generalized HVT} (GHVT) to refer to a particular implementation of this suggestion, viz. one that takes the relevant framework to be given by a WPS. As understood in the literature (and in order to reasonably count as a hidden variable theory) a GHVT at the very least needs to include a specification of events, probabilities of such events, and states of affairs. The first two are taken to be $\Sigma$ and $\mu$ of the WPS, respectively, but very little attention has been paid to how `states of affairs' are to be specified. We shall return to this point in Section 5.
% We flag that here we are merely reporting how WPSs have been used to generalize CHVTs in the literature---we certainly do not endorse these proposals or consider them to be well-motivated (indeed we do not even think that the literature has been clear about what a GHVT consists in, as evidenced by the lack of detail and reflection about how `states of affairs' are to be specified). 

The connection between WPSs and Dutch Books emerges when Feintzeig (\citeyear{Feintzeig2014}) and Feintzeig and Fletcher (\citeyear{Feintzeig2017}) attempt to mount a Dutch Book argument against the GHVT implementation of Fine's suggestion.
We describe this in two steps. First, they argue for the following claim:
\begin{quote} \textbf{(Dutch)}
There exists a quantum model $Q$ such that any WPS that reproduces $Q$'s empirical probabilities will violate a formal constraint, which in turn shows that it will be susceptible to a (formal) Dutch Book.
% \footnote{In fact, several different kinds of formal constraints imply formal Dutch Books---we further clarify this point in Section 4.}
\end{quote}
Regardless of whether we are discussing quantum or non-quantum models, we will refer to this type of \textit{formal} constraint-violation as \textit{Dutch Bookability}. We emphasize that this formal constraint-violation can be understood \textit{without} either interpreting a WPS as a hidden variable theory or taking a subjective approach to probability, which is how we shall treat it in the rest of this paper, with the exception of Section 5.  

Second, they interpret the Dutch Book normatively, thereby taking it to provide a successful no-go argument against `noncontextual' GHVTs of a certain kind. 
It would be interesting to provide enough detail to this argument so that one can evaluate it (e.g. how should one understand `states of affairs' for such a GHVT, and what does `noncontextual' mean in this context?). 
But although we take up this subject in Section 5, our main goal is to answer a far more fundamental question concerning (Dutch), whose answer informs our assessment of the normative interpretation of Dutch Bookability. 

\subsection{Our motivation and results}

This more fundamental question is: what is the precise relationship between (Dutch) and (Standard), as well as their theory-independent generalizations?
This question has been left unanswered in the literature, although partial clues indicate that some relationship certainly exists.
For instance, (\cite{Feintzeig2014, Feintzeig2017}) argue by contradiction to show that the Kochen-Specker theorem implies that there exist quantum models which are formally Dutch Bookable. 
However, this leaves unresolved the more general question of how (Standard)---which includes non-Kochen-Specker models such as the Hardy and Bell models---is related to (Dutch), as well as the task of providing a direct and explicit account of this relationship. 
% We note that even independently of (\cite{Feintzeig2014, Feintzeig2017}), one should expect a general correspondence between (Dutch) and (Standard), because de Finetti's (\citeyear{deFinetti-ToP:1974}) presentation of Dutch Books essentially describes the non-Dutch Bookable constraint as characterizing the probabilities that can be reproduced by a CHVT, and it is precisely such a reproduction that is ruled out by (Standard). 
% Thus, one should expect non-Kochen-Specker (Standard) models to also be Dutch Bookable. We'll return to this point in our discussion of the normativity of (Dutch) arguments in Section \ref{sec:normative}.

A further motivation for pursuing our topic is to provide an explanation of why the particular kind of Dutch Bookability that appears in (\cite{Feintzeig2014,Feintzeig2017}) stems from such a dramatic violation of additivity, in which the sample space is covered by a finite number of measure zero sets. Why should the Dutch Books of quantum theory be so extreme? 
In fact, the existence of a correspondence between (Dutch) and (Standard) should lead one to expect that not all such Dutch Books are so extreme, because---intuitively---not all (Standard) models are as `extreme' as Kochen-Specker models in their deviation from classical probability.

Our strategy for gaining cognitive control over these issues and resolving them is to harness the power of AB's framework for contextuality by constructing a `salient/physical feature-preserving' map taking any model from AB's framework to appropriate models in the purely formal WPS framework in which Dutch Bookability has been articulated.
The use of AB's framework is particularly illuminating in this regard because (as we will see in Section 3.1) (i) all quantum models and many non-quantum models can be represented within it; and (ii) it provides an abstract definition of contextuality and classifies contextual models into three different tiers, viz.---going from strongest to weakest---strong, logical, and probabilistic.
Thus, if one succeeds in constructing an adequate map from AB's framework to WPSs, one immediately obtains the WPS version of (i) and (ii).

Roughly speaking, our main result (see Theorem 1 and Proposition 1) can be summarized as saying that the WPS version of contextuality just is Dutch Bookability! 
In particular, our map sends strongly contextual empirical models (such as the Kochen-Specker model) to WPSs that maximally violate subadditivity, logically contextual empirical models (such as the Hardy model) to WPSs that violate subadditivity, and probabilistically contextual empirical models (such as the Bell model) to WPSs that violate additivity.
This result then allows us to resolve the puzzle about why the Dutch Bookability that appears in (\cite{Feintzeig2014, Feintzeig2017}) is so dramatic: Kochen-Specker models are not the only kind of quantum models giving rise to Dutch Books; other (Standard) models such as the Hardy and Bell models do so as well. However, Kochen-Specker models result in such a dramatic form of Dutch Bookability precisely because they witness the strongest form of contextuality. 
The details of these results are carefully explained in the next two sections of the paper, but the reader who is only interested in the broader conceptual story can now skip ahead to Section 5, where we discuss Feintzeig and Fletcher's normative interpretation of (formal) Dutch Bookability as a no-go argument against GHVTs.

\section{Representing empirical models}

In Section 3.1, we provide an elementary account of AB's framework for Gluing Inconsistency that is accessible to readers with basic knowledge of set theory (like AB, we will only discuss finite sets).
We also use a diagrammatic visualization (see Figure 3.1) to convey why AB's framework is so illuminating for our purposes, viz. it classifies various physical models (including quantum models) into a hierarchy of different strengths of contextuality.  
Section 3.2 introduces the purely formal framework of Weak Probability Spaces (WPSs) that will be used to formulate Dutch Bookability, and explains how the salient features of AB's framework can be translated into those of a WPS. 
In other words, we are setting up the machinery that is required for Section 4 to demonstrate that Dutch Bookability is the WPS version of AB's contextuality. 

\subsection{Gluing Inconsistency: Empirical models}\label{sec:gluing_formal}

%We now describe AB's (\citeyear{AB}) framework for Gluing Inconsistency.
%Although Abramsky and collaborators (\citeyear{AB,Abramsky2015}) deploy the full machinery of `sheaves' and `presheaves', we will not need to explicitly define these objects because we will only use their most elementary properties, which we will describe directly. Like AB, we will only discuss finite sets for the sake of simplicity.
%formalism we use for the remainder of the paper. In section \ref{sec:sheaf_formal}, we will illustrate how Abramsky and Brandenburger use obstructions to the existence of global sections of sheaves to characterize (Standard). In section \ref{sec:dutch_formal}, we will show how Feintzeig and Fletcher use the properties of weak probability spaces to characterize (Dutch).

The first notion that we will need to set up AB's framework is that of  an \emph{empirical scenario}, viz. a triple $(X,\mathcal{M},O)$, where $X$ is a set of measurements, $\mathcal{M}$ is the set of maximal contexts in the power-set $P(X)$, and $O$ is the set of possible outcomes for each measurement (recall that a `context' is a set of co-measurable/compatible measurements). We will also find it useful to refer to the set of all (possibly non-maximal) contexts, denoted $\M'$. 
Empirical scenarios can be represented by `bundle diagrams': e.g. in Fig. \ref{fig:bundles}, vertices of the `base' represent $X = \{ a , b, a', b' \}$, edges of the `base' represent $\M = \{ \{a,b\}, \{b, a'\}, \{a', b'\}, \{b', a\}\}$, and vertices over the `base' represent $O=\{0, 1\}$.

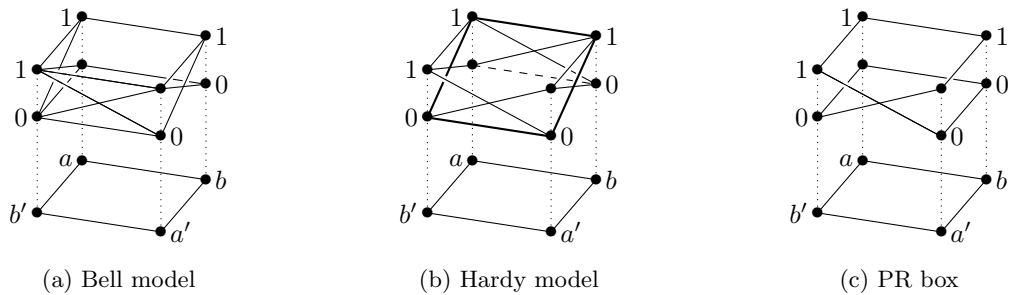
\begin{figure}[!htb]
    \centering
    %%%%%%%%%%%%%%
    % Bell model %
    %%%%%%%%%%%%%%
    \begin{subfigure}{.3\textwidth}
      \centering
        \tdplotsetmaincoords{65}{110}
        \begin{tikzpicture}[tdplot_main_coords,scale=0.7]
        \pgfmathsetmacro{\factor}{2.5};
        \pgfmathsetmacro{\inter}{1.5};
        \coordinate [label=left:$a$] (a) at (0*\factor,0*\factor,0);
        \coordinate [label=right:$b$] (b) at (0*\factor,1*\factor,0);
        \coordinate [label=left:$b'$] (b') at (1*\factor,0*\factor,0);
        \coordinate [label=right:$a'$] (a') at (1*\factor,1*\factor,0);
        
        \coordinate [label=left:$1$] (a1) at (0*\factor,0*\factor,3);
        \coordinate [label=right:$1$] (b1) at (0*\factor,1*\factor,3);
        \coordinate [label=left:$1$] (b'1) at (1*\factor,0*\factor,3);
        \coordinate [label=right:] (a'1) at (1*\factor,1*\factor,3);
        
        \coordinate [label=left:] (a0) at (0*\factor,0*\factor,2);
        \coordinate [label=right:$ 0$] (b0) at (0*\factor,1*\factor,2);
        \coordinate [label=left:$ 0$] (b'0) at (1*\factor,0*\factor,2);
        \coordinate [label=right:$ 0$] (a'0) at (1*\factor,1*\factor,2);
        
        \foreach \Point in {(a),(b),(a'),(b'),(a0),(b0),(a'0),(b'0),(a1),(b1),(a'1),(b'1)}{
            \node at \Point {\textbullet};
        }
        
        % base space
        \draw[-, opacity=1] (a)--(b)--(a')--(b')--cycle;
        
        % stalks
        \draw[dotted, opacity=1] (a)--(a0)--(a1);
        \draw[dotted, opacity=1] (b)--(b0)--(b1);
        \draw[dotted, opacity=1] (a')--(a'0)--(a'1);
        \draw[dotted, opacity=1] (b')--(b'0)--(b'1);
        
        % ab sections
        \draw[-, opacity=1, name path = a0b0] (a0)--(b0);
        % \draw[-, opacity=1] (a0)--(b1);
        % \draw[-, opacity=1] (a1)--(b0);
        \draw[-, opacity=1] (a1)--(b1);
        
        % ab' sections
        \draw[-, opacity=1, name path = a0b'0] (a0)--(b'0);
        \draw[-, opacity=1] (a0)--(b'1);
        \draw[-, opacity=1, name path = a1b'0] (a1)--(b'0);
        \draw[-, opacity=1] (a1)--(b'1);
        
        % a'b sections        
        \draw[-, opacity=0, name path = a'0b1] (a'0)--(b1);
        \draw[-, opacity=1] (a'1)--(b0);
        \draw[-, opacity=0, name path = a'1b1] (a'1)--(b1);
        
        \path [name intersections={of=a'1b1 and a0b0,by=intera}];
        \path [name intersections={of=a'0b1 and a0b0,by=interb}];
        \filldraw [white] (intera) circle (\inter pt);
        \filldraw [white] (interb) circle (\inter pt);
        \draw[-] (a'1)--(b1);
        \draw[-] (a'0)--(b1);
        
        % a'b' sections
        \draw[-, opacity=1] (a'0)--(b'0);
        \draw[-, opacity=1, name path = a'0b'1] (a'0)--(b'1);
        \draw[-, opacity=1, name path = a'1b'0] (a'1)--(b'0);
        \draw[-, opacity=1, name path = a'1b'1] (a'1)--(b'1);
        
        \path [name intersections={of=a0b'0 and a'1b'1,by=inter1}];
        \path [name intersections={of=a0b'0 and a'0b'1,by=inter2}];
        \path [name intersections={of=a1b'0 and a'1b'1,by=inter3}];
        \path [name intersections={of=a1b'0 and a'0b'1,by=inter4}];
        \filldraw [white] (inter1) circle (\inter pt);
        \filldraw [white] (inter2) circle (\inter pt);
        \filldraw [white] (inter3) circle (\inter pt);
        \filldraw [white] (inter4) circle (\inter pt);
        \draw[-, opacity=1] (a'1)--(b'1);
        \draw[-, opacity=1] (a'0)--(b'1);
        \end{tikzpicture}
      \caption{Bell model}
      \label{fig:Bell}
    \end{subfigure}% 
    ~
    %%%%%%%%%%%%%%%
    % Hardy model %
    %%%%%%%%%%%%%%%
    \begin{subfigure}{.3\textwidth}
      \centering
        \tdplotsetmaincoords{65}{110}
        \begin{tikzpicture}[tdplot_main_coords,scale=0.7]
        \pgfmathsetmacro{\factor}{2.5};
        \pgfmathsetmacro{\inter}{1.5};
        \coordinate [label=left:$a$] (a) at (0*\factor,0*\factor,0);
        \coordinate [label=right:$b$] (b) at (0*\factor,1*\factor,0);
        \coordinate [label=left:$b'$] (b') at (1*\factor,0*\factor,0);
        \coordinate [label=right:$a'$] (a') at (1*\factor,1*\factor,0);
        
        \coordinate [label=left:$1$] (a1) at (0*\factor,0*\factor,3);
        \coordinate [label=right:$1$] (b1) at (0*\factor,1*\factor,3);
        \coordinate [label=left:$1$] (b'1) at (1*\factor,0*\factor,3);
        \coordinate [label=right:] (a'1) at (1*\factor,1*\factor,3);
        
        \coordinate [label=left:] (a0) at (0*\factor,0*\factor,2);
        \coordinate [label=right:$ 0$] (b0) at (0*\factor,1*\factor,2);
        \coordinate [label=left:$ 0$] (b'0) at (1*\factor,0*\factor,2);
        \coordinate [label=right:$ 0$] (a'0) at (1*\factor,1*\factor,2);
        
        \foreach \Point in {(a),(b),(a'),(b'),(a0),(b0),(a'0),(b'0),(a1),(b1),(a'1),(b'1)}{
            \node at \Point {\textbullet};
        }
        
        % base space
        \draw[-, opacity=1] (a)--(b)--(a')--(b')--cycle;
        
        % stalks
        \draw[dotted, opacity=1] (a)--(a0)--(a1);
        \draw[dotted, opacity=1] (b)--(b0)--(b1);
        \draw[dotted, opacity=1] (a')--(a'0)--(a'1);
        \draw[dotted, opacity=1] (b')--(b'0)--(b'1);
        
        % ab sections
        \draw[dashed, opacity=1, name path = a0b0] (a0)--(b0);
        \draw[-, opacity=1, name path = a0b1] (a0)--(b1);
        \draw[-, opacity=1, name path = a1b0] (a1)--(b0);
        \draw[thick, opacity=1, name path = a1b1] (a1)--(b1);
        
        % ab' sections
        % \draw[-, opacity=1] (a0)--(b'0);
        \draw[-, opacity=1] (a0)--(b'1);
        \draw[thick, name path = a1b'0] (a1)--(b'0);
        \draw[-, opacity=1] (a1)--(b'1);
        
        % a'b sections
        % \draw[-, opacity=1] (a'0)--(b0);
        \draw[-, opacity=0, name path = a'0b1] (a'0)--(b1);
        \draw[-, opacity=1] (a'1)--(b0);
        \draw[-, opacity=0, name path = a'1b1] (a'1)--(b1);
        
        \path [name intersections={of=a'1b1 and a0b0,by=inter1}];
        \path [name intersections={of=a'1b1 and a1b0,by=inter2}];
        \path [name intersections={of=a'0b1 and a0b0,by=inter3}];
        \path [name intersections={of=a'0b1 and a1b0,by=inter4}];
        \filldraw [white] (inter1) circle (\inter pt);
        \filldraw [white] (inter2) circle (\inter pt);
        \filldraw [white] (inter3) circle (\inter pt);
        \filldraw [white] (inter4) circle (\inter pt);
        \draw[-, opacity=1] (a'1)--(b1);
        \draw[thick] (a'0)--(b1);
        
        % a'b' sections
        \draw[thick] (a'0)--(b'0);
        \draw[-, opacity=0, name path = a'0b'1] (a'0)--(b'1);
        \draw[-, opacity=1] (a'1)--(b'0);
        % \draw[-, opacity=1] (a'1)--(b'1);
        
        \path [name intersections={of=a1b'0 and a'0b'1,by=inter5}];
        \filldraw [white] (inter5) circle (\inter pt);
        \draw[-, opacity=1] (a'0)--(b'1);
        \end{tikzpicture}
      \caption{Hardy model}
      \label{fig:Hardy}
    \end{subfigure}%
    ~
    %%%%%%%%%%
    % PR box %
    %%%%%%%%%%
    \begin{subfigure}{.3\textwidth}
      \centering
        \tdplotsetmaincoords{65}{110}
        \begin{tikzpicture}[tdplot_main_coords,scale=0.7]
        \pgfmathsetmacro{\factor}{2.5};
        \pgfmathsetmacro{\inter}{1.5};
        \coordinate [label=left:$a$] (a) at (0*\factor,0*\factor,0);
        \coordinate [label=right:$b$] (b) at (0*\factor,1*\factor,0);
        \coordinate [label=left:$b'$] (b') at (1*\factor,0*\factor,0);
        \coordinate [label=right:$a'$] (a') at (1*\factor,1*\factor,0);
        
        \coordinate [label=left:$1$] (a1) at (0*\factor,0*\factor,3);
        \coordinate [label=right:$1$] (b1) at (0*\factor,1*\factor,3);
        \coordinate [label=left:$1$] (b'1) at (1*\factor,0*\factor,3);
        \coordinate [label=right:] (a'1) at (1*\factor,1*\factor,3);
        
        \coordinate [label=left:] (a0) at (0*\factor,0*\factor,2);
        \coordinate [label=right:$ 0$] (b0) at (0*\factor,1*\factor,2);
        \coordinate [label=left:$ 0$] (b'0) at (1*\factor,0*\factor,2);
        \coordinate [label=right:$ 0$] (a'0) at (1*\factor,1*\factor,2);
        
        \foreach \Point in {(a),(b),(a'),(b'),(a0),(b0),(a'0),(b'0),(a1),(b1),(a'1),(b'1)}{
            \node at \Point {\textbullet};
        }
        
        % base space
        \draw[-, opacity=1] (a)--(b)--(a')--(b')--cycle;
        
        % stalks
        \draw[dotted, opacity=1] (a)--(a0)--(a1);
        \draw[dotted, opacity=1] (b)--(b0)--(b1);
        \draw[dotted, opacity=1] (a')--(a'0)--(a'1);
        \draw[dotted, opacity=1] (b')--(b'0)--(b'1);
        
        % ab sections
        \draw[-, opacity=1, name path = a0b0] (a0)--(b0);
        % \draw[-, opacity=1] (a0)--(b1);
        % \draw[-, opacity=1] (a1)--(b0);
        \draw[-, opacity=1] (a1)--(b1);
        
        % ab' sections
        \draw[-, opacity=1, name path = a0b'0] (a0)--(b'0);
        % \draw[-, opacity=1] (a0)--(b'1);
        % \draw[-, opacity=1] (a1)--(b'0);
        \draw[-, opacity=1] (a1)--(b'1);
        
        % a'b sections
        \draw[-, opacity=1] (a'0)--(b0);
        % \draw[-, opacity=1] (a'0)--(b1);
        % \draw[-, opacity=1] (a'1)--(b0);
        \draw[-, opacity=0, name path = a'1b1] (a'1)--(b1);
        
        \path [name intersections={of=a0b0 and a'1b1,by=inter1}];
        \filldraw [white] (inter1) circle (\inter pt);
        \draw[-, opacity=1] (a'1)--(b1);
        
        % a'b' sections
        % \draw[-, opacity=1] (a'0)--(b'0);
        \draw[-, opacity=1, name path = a'0b'1] (a'0)--(b'1);
        \draw[-, opacity=1] (a'1)--(b'0);
        % \draw[-, opacity=1] (a'1)--(b'1);
        
        \path [name intersections={of=a0b'0 and a'0b'1,by=inter2}];
        \filldraw [white] (inter2) circle (\inter pt);
        \draw[-, opacity=1] (a'0)--(b'1);
        \end{tikzpicture}
      \caption{PR box}
      \label{fig:PR}
    \end{subfigure}      
    \caption{Bundle diagrams for the support of $e$.}
    \label{fig:bundles}
\end{figure}

Given an empirical scenario, we can define (measurement) events by means of the \textit{event sheaf} $\mathcal{E}$, which maps each $U \subseteq X$ to the set of \textit{events} $ \mathcal{E}(U):=O^U$ (i.e. an event is a function from $U$ to $O$), and which maps inclusions $U \subseteq U'$ to the obvious restriction $s \mapsto s|_U$, where $s \in \mathcal{E}(U')$. 
% \footnote{Where appropriate, we will use the symbol $|$ to refer either to set-theoretic restriction or to the sort of restriction specified here---the restriction along a morphism as specified by a presheaf (\cite[p. 25]{MacLane2012}).} 
When we want to emphasize the geometric character of events in the AB framework, we will refer to events as \textit{sections} over $U$; if $U=X$, we say the sections are \emph{global}, and if $U\neq X$, we say the sections are \emph{local}. 
This geometric interpretation is illustrated in Fig. \ref{fig:bundles}, where sections are represented as line segments over the base. For instance, in Fig. \ref{fig:Bell}, there are 4 local sections over the context $\{b', a' \}$. 

We now highlight an obvious but important feature of events in the AB framework. Let $U\subseteq X$ be covered by $\{ U_i \}$. We say that a family of sections $\{ s_i \in \E(U_i) \}$ is \textit{compatible} just in case $s_i|_{U_i \cap U_j} = s_j|_{U_i \cap U_j}$ for all $i,j$.
The `sheaf property' of $\mathcal{E}$ amounts to the following unremarkable fact: for every compatible family $\{ s_i \}$, there is a unique `larger' section $s \in \E(U)$ that restricts to $s_i$ on each $U_i$, viz. the function $s:U \rightarrow O$ that is constructed piecewise from $\{ s_i \}$. In particular, when $U=X$ then the family $\{s_i \}$ gives rise to a unique global section (for instance, in Fig. 3.1b, the family of thick edges over the base glues together to form a unique global section).

Having defined events, we now proceed to define the probabilities of these events.
Let a probability distribution on a finite set $A$ be a function $d: A \rightarrow [0,1]$ whose sum over elements of $A$ is normalized to $1$; we will use $\D(A)$ to denote the set of probability distributions on $A$.
AB describe probabilities on the events of a measurement scenario by means of a `pre-sheaf' $\D\E$. 
We will only need two features of this object. First, $\D\E$ maps each context $U \in \M'$ to the set of distributions $\D\E(U)$ on the events (sections) over $U$; we will use $e_U \in \D\E(U)$ to denote a distribution on this set. 
% \footnote{Note that we could also have used the notation $\D (\E (U))$ to emphasize that this is a set of distributions on $\E (U)$.}
Second, let $U, U' \in \M'$ such that $U' \supset U$ and let $A_s$ be the set of all $r\in \E(U')$ such that $r|_U =s \in \E(U)$. $\D\E$ defines a notion of `distribution marginalization', viz. $e_{U'}|_{U} (s) := \Sum_{r \in A_s} e_{U'} (r)$ for $s\in\E(U)$. We say that a distribution $e_{U'}\in \mathcal{DE}(U')$ `restricts' to a distribution $e_U\in\mathcal{DE}(U) $ if
\beq\label{eq:marg1}
e_U=e_{U'}|_{U};
\eeq
\noindent
note that there is always some $e_U$ for which this is the case.
% For the purposes of our discussion in Section 3.2, it will be important to note that AB's `events' and `distributions' constitute algebraic versions of set-theoretic measure spaces in which events over a context are mutually exclusive (thus, the framework has no need to explicitly represent conjunctions of events over the same context).

%of distributions on events (in a context) assumes that the relevant 
%SAY SOMETHING HERE ABOUT HOW THIS IS ALGEBRAIC AND APPLIES TO MUTUALLY EXCLUSIVE EVENTS.ALSO CLASSICALITY
%[Note that each $d_U$ in $\mathcal{DE}(U)$ assigns measures only to events in $\mathcal{E}(U)$---and not to negations, conjunctions, or disjunctions of them. Thus, the following formal weak classicality condition is satisfied]

We may now define the notion of an `empirical model' that we earlier said would serve as a theory-independent generalization of a quantum model. 
Given an empirical scenario $(X, \M, O)$ and the pre-sheaf $\D\E$, an \textit{empirical model} $e := \{ e_C \}_{C \in \M}$ is a family of distributions that satisfies the following compatibility condition: for all $C,C' \in \M$,
\beq\label{ecompatible}
e_C|_{C \cap C'} = e_{C'}|_{C\cap C'}.
\eeq
AB (\citeyear{AB}) show that, by means of a standard prescription, any quantum model induces a corresponding empirical model.
For instance, the probabilities of the Bell model define a family $\{ e_C \}$ for which the compatibility condition (\ref{ecompatible}) simply amounts to the statement that model's correlations satisfy the `no-signaling' condition (for this reason, (\ref{ecompatible}) is also referred to in the literature as a \textit{generalized no-signaling} condition).
Furthermore, AB show that the class of empirical models is much larger than just the quantum models: it includes various important \textit{non-quantum} models which have been used as `foils' in the foundations literature, such as the PR Box (\cite{Popescu2014}) and generalized Kochen-Specker scenarios like the Specker Triangle (\cite{Spekkens2011}).

The notion of an empirical model allows us to define a three-tiered \textit{hierarchy of contextuality} that corresponds to three different strengths of `gluing failure'.
%And as we are about to see, such gluing failure amounts to the inconsistency of a certain type of HVT with 
The first/strongest/top tier and the second/middle tier are purely `possibilistic' in the sense that their definitions only use the coarse-graining of $e$'s probabilities into `possible' (non-zero probability) and `impossible' (zero probability).
We now define this sector of \textit{possibilistic contextuality}. 
First, we say that the support of an empirical model $e$ is the union of the supports $ \text{supp}(e_C)$ (the set of sections that $e_C$ assigns non-zero probabilities) for each distribution $e_C$. Second, we say that a global (event) section $s \in \E (X)$ is consistent with the support of $e$ just in case $s|_{C} \in \text{supp}(e_C)$ for all $C$; otherwise we say that $s$ is inconsistent with the support of $e$.
\begin{itemize}
\item Tier 1: $e$ is \textit{strongly contextual} iff all its global (event) sections are inconsistent with the support of $e$.
Well-known quantum examples of strong contextuality are provided by all (quantum) Kochen-Specker models and the GHZ model.
Well-known non-quantum examples are provided by the PR Box and the Specker Triangle. 
% \footnote{It is known that the class of strongly contextual empirical models properly contains the class of AvN models, which in turn properly contains generalized Kochen-Specker models; see e.g. (\cite{CCP}).}

\item Tier 2: $e$ is \textit{logically contextual} iff at least one of its global (event) sections is inconsistent with the support of $e$. Clearly, strong contextuality implies logical contextuality, but not vice versa (and so logical contextuality is in this sense `weaker' than strong contextuality). 
A well-known quantum example is provided by Hardy's model (\cite{Hardy1993}).
\end{itemize}

The geometric/topological essence of possibilistic contextuality can easily be visualized by means of `bundle diagrams' of (event) sections in the support of $e$, as depicted in Fig. \ref{fig:bundles}.
In Fig. \ref{fig:PR} we see that the PR Box is strongly contextual because no section over a maximal context can be extended to a global section of the bundle (indeed, the topological structure of this bundle is a combinatorial version of a `Mobius strip', which has no global sections).
Similarly, in Fig. \ref{fig:Hardy}, we see that the Hardy model is logically contextual because there exists a section (the dotted edge) which cannot be extended to a global section of the bundle. However, the Hardy model is \textit{not} strongly contextual, because it also contains local sections (e.g. the thick edges) that extend to global sections of the bundle.

%be Possibilistically contextual phenomena admit of an intuitive graphical depiction that we call `support diagrams'.
%The Bell model, the support of which is pictured in Figure \ref{fig:Bell}, exemplifies probabilistic contextuality. On the middle rung, some sections in the support extend to global sections, but others do not. In this case, we say $\omega$ is \emph{logically contextual}. Finally, on the last rung, no sections in the support extend to global sections; in such a case, we say $\omega$ is \emph{strongly contextual}.

The third and lowest tier of the hierarchy is a failure of gluing that is defined at the level of \textit{probabilities}:
\begin{itemize}
\item Tier 3: $e:= \{ e_C \}_{C \in \M}$ is \textit{probabilistically contextual} iff there does not exist a global (or joint) distribution $e_X$ that restricts to $e_C$ on each maximal context $C$.
\end{itemize}
Clearly, any empirical model that is possibilistically contextual is also probabilistically contextual---but the converse does not hold.
%(because if there is a global event section that is inconsistent with the support of $e$, then there can be no global distribution that restricts to produce the correct probabilities on local sections). 

A well-known example of probabilistic contextuality is the Bell model, whose bundle diagram is depicted in Fig. \ref{fig:Bell}. 
Notice that all of the bundle's sections extend to global sections, and thus it is not possibilistically contextual.
But as Bell famously proved, the model is nonetheless probabilistically contextual, because the family of probability distributions (determined by Alice and Bob's different measurement settings) does not give rise to the joint probability distribution required for the existence of a local (NC) HVT. 
More generally, AB (\citeyear{AB}) show that an empirical model $e$ has a global distribution iff there exists a (NC) `factorizable' HVT that reproduces the probabilities of $e$: this result thus constitutes a theory-independent generalization of (Standard).\footnote{AB define `factorizability' as follows: consider the HVT specified by some distribution $h_{\Lambda}\in\D_R(\Lambda)$ over ontic states $\lambda\in\Lambda$, where ontic states induce response functions $h_C^{\lambda}\in \D_R \E (C)$ such that $e_C = \sum_{\lambda\in\Lambda} h_C^{\lambda}\cdot h_{\Lambda}(\lambda) $ for all $C\in\M$.
This HVT is \emph{factorizable} (\cite[Section 8]{AB}) iff, for any $\lambda\in\Lambda$ and for any $s\in\E(C)$ given any $C\in\M$, $
h_C^{\lambda} (s) = \prod_{m\in C} h_C^{\lambda}|_{\{m\}} \left (s|_{\{m\}}\right )$.}

%Factorizability is generally taken to encode Bell's notion of locality \citep{Brunner2014}. The notion of %factorizability employed by AB is that pertaining to a specific sort of HVT---one where a preparation of a system induces a distribution $h_{\Lambda}\in\D_R(\Lambda)$ over ontic states $\lambda\in\Lambda$ and ontic states induce distributions $h_C^{\lambda}\in \D_R \E (C)$ encoding the probabilities for observing the event $s$ given the system is in state $\lambda$ after preparation and the measurements in $C$ are performed on it. Moreover: distributions $h_C^{\lambda}$ are required to be compatible in the sense of the `sheaf condition' and recover empirical probabilities, viz. $e_C = \sum_{\lambda\in\Lambda} h_C^{\lambda}\cdot h_{\Lambda}(\lambda) $ for all $C\in\M$. Such a HVT is \emph{factorizable} \citep[Section 8]{AB} iff, for any $\lambda\in\Lambda$ and for any $s\in\E(C)$ given any $C\in\M$,
%$$
%h_C^{\lambda} (s) = \prod_{m\in C} h_C^{\lambda}|_{\{m\}} \left (s|_{\{m\}}\right ).
%$$
%} 
%As a special case, if $e$ is strongly or logically contextual, then there does not exist an (NC, OD) HVT that reproduces the probabilities of $e$.
%---this result provides a theory-independent generalization of quantum contextuality arguments. 
In sum, the hierarchy of contextuality is:
\beq\label{eq:contexthier}
\text{Strong} \implies \text{Logical} \implies \text{Probabilistic}
\eeq
In the next subsection, we explain how the structural features that allow us to define this hierarchy can be represented within the WPS framework.

\subsection{Weak Probability Spaces: Representations of empirical models}\label{sec:dutch_formal}

Feintzeig and Fletcher (\citeyear{Feintzeig2017}) assume that Generalized HVTs (GHVTs) will at least have the structure of a Weak Probability Space (WPS). Thus, in order to demonstrate that GHVTs are Dutch Bookable, they provide a method of representing a quantum Kochen-Specker model within the WPS framework.
The main task of this section will be to obtain a theory-independent generalization of this method by constructing a WPS representation of an empirical model. 

Recall that a probability space is a triple $(Y, \Sigma, \mu)$ where $\Sigma$ is an algebra of events (viz. a $\sigma$-field on $Y$), and $\mu: \Sigma \rightarrow [0,1]$ is a probability measure.\footnote{Recall that a (probability) measure is an additive function $\mu:\Sigma\to \mathbb{R}_{\geq 0}$ such that $\mu (\emptyset)=0$ and $\mu (Y)=1$. A $\sigma$-field $\Sigma$ on $Y$ is a collection of subsets of $Y$ containing $\emptyset$ that is closed under complementation and countable unions and intersections.}
A WPS is also a triple $(Y, \Sigma, \mu)$, but it generalizes a measure space by replacing the algebra of events with a mere set of events $\Sigma \subseteq P(Y)$ and by replacing the measure with a mere set function $\mu: \Sigma \rightarrow R \supseteq  [0,1]$; we will nonetheless find it convenient to refer to this set function as a `measure' when there is no danger of confusion. 
% \footnote{Cf. Definition 5 of (\cite[p. 297]{Feintzeig2017}). We alter the codomain of their definition to the unit interval, since we have no need of the additional generality.}
While this definition of $\Sigma$ is incredibly weak, remember that it only functions as a template: when a WPS is used to represent an empirical model, this template will need to be filled in with the relevant co-measurability relations.

We now proceed to explain how events in AB's framework (`AB-events') can be \textit{represented} in a WPS, and how the WPS events (`$\Sigma$-events') can then be constructed from a proper subset of these representations.
For brevity, we will use `$s_U$' to denote an AB-event section in $\E (U)$.

Let $e$ be an empirical model.
An \textit{event representation} of $e$, denoted $\E^\dag$, consists of a sample space $Y$ and an injective `event transfer map' 
\beq
   \Ebar :  \bigcup_{U \subseteq X}  \E (U) \rightarrow P(Y).
\eeq
that satisfies the following analog of the `sheaf' property for AB-events: for any $U \subseteq U' \subseteq X$,
    \beq\label{eq:setglue}
    \Ebar(s_{U'}) = \bigcap_{\left \{ s_U ~:~  s_{U'}|_{U} = s_U \right \} } \Ebar (s_U) \neq \emptyset.
    \eeq
Recall that $s_{U'}$ is uniquely determined by the AB-events $\{ s_U \}$ that it restricts to; (\ref{eq:setglue}) simply encodes this set-theoretically by representing $s_{U'}$ as the conjunction of $\{ \Ebar (s_U) \}$ (the intersections are required to be non-trivial because any of these events may be possible).

At this juncture, one could impose two further conditions to ensure that an event representation $\E^\dag$ fully encodes the combinatorial structure of the event sheaf $\E$ of $e$, viz. 
\begin{quote}
    \textbf{Strong Mutual Exclusivity} For any distinct $s,s' \in\E (U)$ and any $U\subseteq X$, $\Ebar(s) \cap  \Ebar(s') = \emptyset$.
\end{quote}
\begin{quote}
    \textbf{Exhaustiveness} For any $U\subseteq X$, $\bigcup_{s\in \mathcal{E}(U)} \Ebar(s) = Y$.
\end{quote}
Call an event representation satisfying these additional conditions a \emph{combinatorial event representation}.\footnote{Every empirical model has a minimal combinatorial event representation; we construct one example in the Appendix.} In what follows, we will prove our main result (Theorem \ref{thm:main}) about the relationship between contextuality and Dutch Bookability without assuming that the event representation is combinatorial (this degree of generality allows us to keep our treatment as close to that of Feintzeig and Fletcher as possible).
However, we also show (Props. 1 and 2) that if one assumes the combinatorial structure, one can give a simpler argument for a stronger result, which has converses.

In order to avoid the clutter of repeatedly writing $\Ebar$, we will introduce the notation $S_U := \Ebar (s_U)$ for event representations, as well as $\E^{\dag} (U) := \{ \Ebar (s_U ) : s_U \in \E (U) \}$ for the set of event representations over $U \subseteq X$. 
% This notation provides us with a useful mnemonic, because it reminds us that $S_U$ is the set in $P(Y)$ that represents the AB-event $s_U$. 
%(NB: As we will soon see, $S_U$ need not be a $\Sigma$-event; indeed, it will not be a $\Sigma$-event if $U$ properly contains a maximal context $C$.)

We note that there is a `duality' between AB-events and their WPS representations: if $U \subset U'$ and $s_U = s_{U'}|_{U}$, then by (\ref{eq:setglue}), $S_{U} \supset S_{U'}$.
Furthermore, just as the event sheaf $\E$ comes with a notion of restriction, $\E^{\dag}$ possesses a dual notion of extension, i.e. $S_{U'}|^{U}=S_U$, where $s_U = s_{U'}|_{U}$. And just as every compatible family of AB-event sections $\{s_i \in \E(U_i) \}$ (where $\{ U_i \}$ covers $X$) gives rise to a unique global section $s_X$, any compatible family $\{ S_i \in \E^\dag (U_i) \}$ (i.e. a family that satisfies $S_{i}|^{U_i \cap U_j} = S_{j}|^{U_i \cap U_j}$ for all $i,j$) has a unique intersection $S_X$ such that $S_{X}|^{U_i} = S_i$ for all $i$.

To prepare ourselves for the definition of a WPS representation of $e$, let us consider how co-measurability relations can be faithfully transferred from the AB-events to the set of $\Sigma$-events that we now wish to construct.
First, recall that in the AB framework, `co-measurable events' are sections over sets of measurements that are jointly contained in a context, and `non-co-measurable events' are sections over sets of measurements that are not. 
%On the other hand, consider that we would like the WPS version of the partial algebra approach to satisfy the following constraint: `co-measurable events' form an algebra of sets in $\Sigma$, whereas `non-co-measurable events' are sets in $\Sigma$ which are not elements of the same algebra. 
Thus, in order to transfer the structure of co-measurable events from the AB framework to the WPS framework, we will restrict the domain of $\Ebar$ to AB-events over contexts, thereby defining the map $E: {\bigcup_{U \in \M'}  \E (U)} \rightarrow P(Y)$. 
It is this map that we will now use to construct the event set $\Sigma$ of a WPS representation.

For any context $U \in \M'$, let $\Sigma_U$ denote the algebra generated by the collection of sets $\bigcup_{x\in U} \E^{\dag} (x)$.\footnote{Here and henceforth, we will use $x$ to denote the singleton $\{x\}$ when this meaning is clear from context.} The set of WPS events $\Sigma$ is then defined as the union $\bigcup_{U \in \M'} \Sigma_U$, where $\Sigma_U$ is treated as a collection of sets. Non-co-measurable $\Sigma$-events are thus elements of $\Sigma$ which are not jointly contained in some $\Sigma_U$.

%A \textit{WPS representation} of $e$, denoted $e^\dag$, is a triple $(\E^\dag, \Sigma, \mu)$ such that $(Y, \Sigma, \mu)$ form a WPS, where $\Sigma$ and $\mu$ are constructed as follows. 
%Let $\Sigma_U$ denote the algebra generated by the collection of sets $\bigcup_{x\in U} \E^{\dag} (x)$, where $U \in \M'$ is a context.
%\footnote{Specifically: $\Sigma_U$, treated as a collection of sets, is the $\sigma$-field generated by $\bigcup_{x\in U} \E^{\dag} (x)$.} 
%The set of events $\Sigma$ is then defined as the union $\bigcup_{U \in \M'} \Sigma_U$, where $\Sigma_U$ is treated as a collection of sets (note that while $E$ is injective, this construction shows that it is not surjective).
A \textit{WPS representation} of $e$, denoted $e^\dag$, is a triple $(\E^\dag, \Sigma, \mu)$ such that $(Y, \Sigma, \mu)$ form a WPS that satisfies the following three conditions. First, we require that $\mu$ allows us to define a probability space for each algebra $\Sigma_U$:
\begin{quote}
\textbf{Weak Classicality (WC)}
Let $e^{\dagger}_U:=\mu|_{\Sigma_U}$. For any context $U \in \M'$, $(Y,\Sigma_U, e^{\dagger}_U )$ is a probability space.
\end{quote}
\noindent
Second, in order to ensure that $e^\dag$ accurately encodes the empirical probabilities of $e$, we require:
\begin{quote}
\textbf{Empirical Consistency (EC)}
For any context $U \subset C  \in \M$ (where $C$ is a maximal context) and any AB-event $s\in \E(U)$,
\beq
e^{\dagger}_U( S ) = e_C |_{U} ( s ).
\eeq
\end{quote}
\noindent 
Third, since $e$ treats events over contexts as mutually exclusive, we require:
\begin{quote}
    \textbf{Mutual Exclusivity (ME)}
    For any distinct AB-events $s, s' \in \E(U)$, $e^{\dagger}_U( S \cap S' ) = 0$.
\end{quote}
\noindent
This concludes our construction of a WPS representation $e^{\dagger}$ of an empirical model $e$.\footnote{For the precise relationship between Feintzeig and Fletcher's \emph{weak hidden variable representations} and our \emph{WPS representations}, see the Appendix.}

As one might expect, any WPS representation $e^\dag$ will obey duals of the marginalization condition (\ref{eq:marg1}) and the compatibility condition (\ref{ecompatible}) for $e$. 
To describe the dual of (\ref{eq:marg1}), let $U \subseteq U' \subseteq C$ and $s\in\E(U)$, let $A_S$ be the set of $R\in\E^\dag(U')$ such that $R|^{U} = S$, and define `marginalization for extension' as $e^{\dagger}_{U'} |^{U} ( S) := \Sum_{ R \in A_S } e^{\dagger}_{U'} ( R)$. By (EC), it immediately follows that
\beq\label{eq:marg2}
e^{\dagger}_{U'} |^{U}=e^{\dagger}_{U},
\eeq
which is the dual of (\ref{eq:marg1}).
The dual of (\ref{ecompatible}) is simply the statement that WPS representations satisfy the following compatibility condition: for all maximal contexts $C, C' \in \M$,
\beq\label{eq:mucompatible}
e^{\dagger}_C|^{C \cap C'} = e^{\dagger}_{C'}|^{C \cap C'}.
\eeq

We now have two ways of formalizing the empirical probabilities of a wide class of physical models: AB's sheaf-theoretic approach, on the one hand, and the above `dual' WPS approach, on the other.  
In the next section, we will argue that AB's hierarchy of contextuality (or Gluing Inconsistency) corresponds to a hierarchy of Dutch Bookability in the WPS approach.

\section{The hierarchies} 

We begin this section by situating our result with respect to previous work on hierarchies of constraint-violations leading to Dutch Books, albeit work that has only been done for \textit{abstract} WPSs, and not WPS representations of \textit{physical} models.
Recall that Feintzeig and Fletcher (\citeyear{Feintzeig2017}) show that any WPS representation of a Kochen-Specker (i.e. a quantum) model will violate a constraint (in their terminology: the `no finite null cover' (NFNC) condition), which amounts to a maximal violation of `subadditivity' in a sense that we will make precise below. 
This of course raises the question of why this violation is maximal, and whether there might be a hierarchy of different violations enjoyed by quantum and other physical models.

Feintzeig (\citeyear{Feintzeig2014}) and Wro\'{n}ski and Godziszewski (\citeyear{Wronski2017}) have attempted to place some constraints on this question by reasoning abstractly about WPSs.
Feintzeig (\citeyear{Feintzeig2014}) notes that even if an abstract WPS does not maximally violate subadditivity, it could still exhibit a weaker violation of subadditivity that gives rise to a formal Dutch Book. 
This leads him to conjecture that an abstract WPS satisfies subadditivity iff it is not Dutch Bookable---if this were true, a WPS representation of a physical model might both offer a non-classical alternative to HVTs and avoid Dutch Books.
However, Wro\'{n}ski and Godziszewski (\citeyear{Wronski2017})  show that this conjecture is false by constructing an abstract WPS that is both Dutch Bookable and subadditive.
More importantly, they show that an abstract WPS is Dutch Bookable just in case it can be modeled by classical probability (cf. Section \ref{sec:additivity} below). 

%This shows that from the perspective of abstract WPSs and Dutch Books, the lowest tier is somewhat boring: it is simply classical probability.
We concur with Wro\'{n}ski and Godziszewski (\citeyear{Wronski2017}) that this exchange demonstrates the futility of simultaneously using WPSs as a framework for GHVTs and stipulating that probabilities in these theories must avoid formal Dutch Books. 
% (of \emph{any} strength) in the sense defined in Section \ref{sec:convexity} below.
However, for all this talk of an abstract hierarchy, it could still be the case that no interesting class of physical models gives rise to Dutch Books except by maximally violating subadditivity!
In other words, it might be the case that as regards WPS representations of \textit{physical models}, no discussion of a hierarchy is necessary beyond the original Dutch Book violation demonstrated by Feintzeig and Fletcher (\citeyear{Feintzeig2017}).

In what follows, we will show that this is not the case by arguing for a correspondence between AB's hierarchy of contextuality and a hierarchy of `formal constraint'-violations that lead to Dutch Bookability.
Since Feintzeig (\citeyear{Feintzeig2014}) emphasizes formal constraints with a measure-theoretic flavor, viz. additivity conditions, and the approach of Wro\'{n}ski and Godziszewski (\citeyear{Wronski2017}) suggests convexity constraints, we will discuss both aspects in Sections \ref{sec:additivity} and \ref{sec:convexity} respectively.
% It will be an immediate consequence of our discussion that there are physical models which violate subadditivity, but not maximal subadditivity (or NFNC); and that there are physical models which violate additivity, but not subadditivity.
%Indeed, for all this talk of the `space' in between classical probability and the maximal violation of subadditivity, it could still be the case that physical models only ever give rise to Dutch Books by maximally violating subadditivity!
%This is not the case; in the next few sections, we will explain why the hierarchy is inhabited. 

\subsection{Contextuality and additivity-violation}\label{sec:additivity}

We now introduce the relevant additivity conditions that will be violated in our hierarchy. These violations are always witnessed by the `defect of subadditivity' for some (finite) set $V$ of $\Sigma$-events,
\beq
\mathfrak{a}(V):= \mu \left ( \bigcup V \right ) - \Sum_{a \in V} \mu (a ).
\eeq
A WPS violates subadditivity just in case there is a $V$ such that $\mathfrak{a} (V) > 0$, and it maximally violates subadditivity just in case there is a $V$ such that $\mathfrak{a} (V)=1$. Note that a `finite null cover' of a WPS representation---a finite collection of measure-zero sets covering the sample space $Y$---provides one example of a maximal violation of subadditivity. A WPS violates additivity just in case there is some (finite) set $V$ of \emph{disjoint} $\Sigma$-events such that $\mathfrak{a}(V)\neq 0$.

In order to construct a hierarchy of additivity-violation, we will also need the notion of an `extension' of a WPS $W \equiv (Y, \Sigma, \mu)$, as well as the more refined notion of a monotonic extension. An \textit{extension} of $W$ is a WPS $(Y,\Sigma',\mu')$, where $\Sigma'$ contains the algebra generated by $\Sigma$, and $\mu'|_{\Sigma}=\mu$.\footnote{A classical extension is just an extension where $\mu'$ is a probability measure.} A \textit{monotonic extension} of $W$ is an extension $(Y, \Sigma' , \mu')$ of $W$ for which $\mu': \Sigma' \rightarrow [0,1]$ is a monotonic function, i.e. if $ A \subseteq B$, then $\mu'(A) \leq \mu' (B)$.\footnote{Note that $\mu'$ is really a `fuzzy measure' in the sense of (\cite{Sugeno}).
% There are also deeper theoretical reasons for imposing monotonicity on an extension: when we describe the space of `probabilistic but non-possibilistic' contextuality, imposing subadditivity on the extension and taking $P(Y)$ to be the domain of $\mu'$ yields a Caratheodory `outer measure', which gives us control over when we can construct genuine measures on algebras within $P(Y)$.
}

In the case of \textit{abstract} WPSs, it is easy to construct a hierarchy of additivity-violation that is tantalizingly similar to AB's hierarchy of contextuality.
It is clear that a maximal violation of subadditivity implies a violation of subadditivity, which in turn implies that any of the monotonic extensions of the WPS in question must violate additivity.
%\footnote{Let $\{A_i \}$ be a collection of sets. The inclusion-exclusion principle is the statement that $\mu' ( \bigcup_{i=1}^{n} A_i) = (-1)^{n-1} \mu' (\bigcap_{i=1}^{n} A_i )
Given this observation, one may naturally expect the following hierarchy for WPS representations of empirical models: strong contextuality will correspond to a maximal violation of subadditivity; logical contextuality will correspond to a violation of subadditivity; and probabilistic contextuality will correspond to a violation of additivity for monotonic extensions of $e^\dag$.
The following theorem shows that this expectation is realized.

\begin{theorem}\label{thm:main}
 Let $e$ be an empirical model and let $e^\dag$ be any of its WPS representations.
	\begin{enumerate}
	    \item If $e$ is strongly contextual, then $e^{\dagger}$ maximally violates subadditivity.
	    \item If $e$ is logically contextual, then $e^{\dagger}$ violates subadditivity.
	    \item If $e$ is probabilistically contextual, then any monotonic extension of $e^{\dagger}$ violates additivity.%\footnote{For this statement, we make the natural assumption that if $e^\dag$ satisfies sub-additivity, then we will only consider its extensions which also satisfy subadditivity. This is because we only wish to consider extensions which are designed to capture the properties of $e^\dag$.} 
	\end{enumerate}
	
	%Furthermore, in each of these cases, the violations stems from a collection of sets in the algebra generated by the $\Sigma$-events of $e^\dag$.
\end{theorem}

In fact, if one makes the stronger assumption that $e^\dagger$ is \textit{combinatorial}, meaning that its event representation is combinatorial (cf. Section \ref{sec:dutch_formal}), then one obtains an \textit{equivalence} between each tier of contextuality and a corresponding condition that witnesses additivity-violation. We now state these conditions and the equivalence result. 

Given that strong and logical contextuality only concern the possibilistic data of an empirical model $e$, we expect the corresponding (sub)additivity-violations to depend only on the possibilistic data of $e^\dag$. To characterize such violations, we will introduce the notion of an \textit{additive cover}, viz. a finite set $A \subseteq \Sigma$ of mutually disjoint $\Sigma$-events that is additive (so $\mathfrak{a}(A)=0$) and that covers the sample space $Y$.
Now say $e^\dag$ \textit{strongly $\V$-violates subadditivity} just in case a finite set $V$ of measure-zero elements in $\V\subseteq \Sigma$ covers \emph{all} the non-measure-zero elements of an additive cover $A\subseteq \V$ (and so $\mathfrak{a}(V) = 1$). Additionally, say $e^\dag$ \textit{logically $\V$-violates subadditivity} just in case a finite set $V$ of measure-zero elements in $\V\subseteq \Sigma$ covers \emph{some} non-measure-zero element $a$ of an additive cover $A\subseteq \V$ (and so $\mathfrak{a}(V\cup A\setminus \{a\}) >0$).
% Here, $V \cup A\setminus \{a\}$ witnesses a (possibly non-maximal) violation of subadditivity.

Finally, in order to characterize additivity-violations stemming from the probabilistic data of $e^\dag$, we turn to monotonic extensions. Say \emph{$e^\dag$ $\V$-violates additivity} just in case the algebra generated by $\V$ contains a collection $V$ of disjoint sets that violates additivity in any monotonic extension of $e^\dag$ (and so $\mathfrak{a}({V})\neq 0$ in these extensions).

In the following proposition, the above conditions are witnessed by the set $\V_{\M}$ of $\Sigma$-events corresponding to AB-events over maximal contexts (i.e. $\V_{\M} := \{S : S\in \E^\dag (C), C\in\M \}$).
\begin{prop}\label{prop:1}
Let $e$ be an empirical model and let $e^\dagger$ be any of its combinatorial WPS representations.
\begin{enumerate}
	    \item $e$ is strongly contextual iff $e^{\dagger}$ strongly $\V_{\M}$-violates subadditivity.
	    \item $e$ is logically contextual iff $e^{\dagger}$ logically $\V_{\M}$-violates subadditivity.
	    \item $e$ is probabilistically contextual iff $e^\dag$ $\V_{\M}$-violates additivity.
	\end{enumerate}
\end{prop}
\noindent
Note that these violations of additivity are \emph{invariant}: if \emph{one} combinatorial WPS representation of $e$ (strongly or logically) $\V_{\M}$-violates (sub)additivity, then so do \emph{all} combinatorial WPS representations of $e$. There is no apparent reason that this need be the case for Theorem \ref{thm:main}, where we consider \emph{arbitrary} WPS representations (and one may well consider this to be another sort of pathology of these structures).

The proof of Theorem \ref{thm:main} involves some set-theoretic yoga. We relegate the details to the Appendix, but it is possible to briefly convey an intuitive idea of how the argument works. First, observe that $Y$ contains two measure-zero collections of $\Sigma$-events which do not represent any AB-events, viz. the collection of `contradictory events' $D_1$ and the collection of `non-existent outcomes' $D_2$ (cf. Eqns (5.1) and (5.2) of the Appendix for the full definitions).
In order to translate the definitions of strong/logical/probabilistic contextuality into a WPS setting, we thus excise $\bigcup(D_1 \cup D_2)$ from $Y$ to obtain $Z$, which represents all the properties that we care about.\footnote{This rough strategy of `excising sets' was applied by Feintzeig and Fletcher (\citeyear{Feintzeig2017}) only to Kochen-Specker models. We go much further, generalizing the strategy using sheaf theory so that it applies to \emph{all} forms of contextuality.}
The translation is then carried out by relying on a key lemma (Lemma \ref{lem:notnull}), which tells us that any $z \in Z$ is contained in some $S_X$ (representing a global section $s_X$).
For instance, this lemma lets us translate `logical contextuality' into the statement that some subset of $Z$ is contained in a collection of measure zero sets, which lets us construct the relevant sub-additivity-violating collection.

\begin{figure}[!htb]
    \centering
    %%%%%%%%%%%%%%%
    % Hardy model %
    %%%%%%%%%%%%%%%
    \begin{subfigure}{.3\textwidth}
      \centering
        \tdplotsetmaincoords{65}{110}
        \begin{tikzpicture}[tdplot_main_coords,scale=0.7]
        \pgfmathsetmacro{\factor}{2.5};
        \pgfmathsetmacro{\inter}{1.5};
        
        \coordinate [label=left:$a_1$] (a1) at (0*\factor,0*\factor,1);
        \coordinate [label=right:$b_1$] (b1) at (0*\factor,1*\factor,1);
        \coordinate [label=left:$b'_1$] (b'1) at (1*\factor,0*\factor,1);
        \coordinate [label=right:] (a'1) at (1*\factor,1*\factor,1);
        
        \coordinate [label=left:] (a0) at (0*\factor,0*\factor,0);
        \coordinate [label=right:$b_0$] (b0) at (0*\factor,1*\factor,0);
        \coordinate [label=left:$b'_0$] (b'0) at (1*\factor,0*\factor,0);
        \coordinate [label=right:$a'_0$] (a'0) at (1*\factor,1*\factor,0);
        
        \foreach \Point in {(a0),(b0),(a'0),(b'0),(a1),(b1),(a'1),(b'1)}{
            \node at \Point {\textbullet};
        }
        
        % ab sections
        \draw[dashed, opacity=1, name path = a0b0] (a0)--(b0);
        \draw[-, opacity=1, name path = a0b1] (a0)--(b1);
        \draw[-, opacity=1, name path = a1b0] (a1)--(b0);
        \draw[thick, opacity=1, name path = a1b1] (a1)--(b1);
        
        % ab' sections
        % \draw[-, opacity=1] (a0)--(b'0);
        \draw[-, opacity=1] (a0)--(b'1);
        \draw[thick, opacity=1, name path = a1b'0] (a1)--(b'0);
        \draw[-, opacity=1] (a1)--(b'1);
        
        % a'b sections
        % \draw[-, opacity=1] (a'0)--(b0);
        \draw[-, opacity=0, name path = a'0b1] (a'0)--(b1);
        \draw[-, opacity=1] (a'1)--(b0);
        \draw[-, opacity=0, name path = a'1b1] (a'1)--(b1);
        
        \path [name intersections={of=a'1b1 and a0b0,by=inter1}];
        \path [name intersections={of=a'1b1 and a1b0,by=inter2}];
        \path [name intersections={of=a'0b1 and a0b0,by=inter3}];
        \path [name intersections={of=a'0b1 and a1b0,by=inter4}];
        \filldraw [white] (inter1) circle (\inter pt);
        \filldraw [white] (inter2) circle (\inter pt);
        \filldraw [white] (inter3) circle (\inter pt);
        \filldraw [white] (inter4) circle (\inter pt);
        \draw[-, opacity=1] (a'1)--(b1);
        \draw[thick] (a'0)--(b1);
        
        % a'b' sections
        \draw[thick, opacity=1] (a'0)--(b'0);
        \draw[-, opacity=0, name path = a'0b'1] (a'0)--(b'1);
        \draw[-, opacity=1] (a'1)--(b'0);
        % \draw[-, opacity=1] (a'1)--(b'1);
        
        \path [name intersections={of=a1b'0 and a'0b'1,by=inter5}];
        \filldraw [white] (inter5) circle (\inter pt);
        \draw[-, opacity=1] (a'0)--(b'1);
        \end{tikzpicture}
      \caption{Hardy model}
      \label{fig:nerveHardy}
    \end{subfigure}%
    ~
    %%%%%%%%%%
    % PR box %
    %%%%%%%%%%
    \begin{subfigure}{.3\textwidth}
      \centering
        \tdplotsetmaincoords{65}{110}
        \begin{tikzpicture}[tdplot_main_coords,scale=0.7]
        \pgfmathsetmacro{\factor}{2.5};
        \pgfmathsetmacro{\inter}{1.5};
        
        \coordinate [label=left:$a_1$] (a1) at (0*\factor,0*\factor,1);
        \coordinate [label=right:$b_1$] (b1) at (0*\factor,1*\factor,1);
        \coordinate [label=left:] (b'1) at (1*\factor,0*\factor,1);
        \coordinate [label=right:] (a'1) at (1*\factor,1*\factor,1);
        
        \coordinate [label=left:] (a0) at (0*\factor,0*\factor,0);
        \coordinate [label=right:] (b0) at (0*\factor,1*\factor,0);
        \coordinate [label=left:$b'_0$] (b'0) at (1*\factor,0*\factor,0);
        \coordinate [label=right:$a'_0$] (a'0) at (1*\factor,1*\factor,0);
        
        \path[name path = b1b'0] (b1)--(b'0);
        \path[name path = a1a'0] (a1)--(a'0);
        \path [name intersections={of=a1a'0 and b1b'0,by=inter}];
        
        \fill[gray!20] (a1)--(b1)--(a'0)--(b'0)--cycle;
        \fill[gray!10] (a1)--(b'0)--(inter)--cycle;
        \foreach \Point in {(b'0),(a1),(b1),(a'0)}{
            \node at \Point {\textbullet};
        }
        \draw[dashed] (b1)--(b'0);
        \draw[thick] (a1)--(b1)--(a'0)--(b'0)--cycle;
        \filldraw [gray!10] (inter) circle (\inter pt);
        \draw[thick] (a1)--(a'0);

        \end{tikzpicture}
      \caption{3-simplex}
      \label{fig:nerve3}
    \end{subfigure}      
    \caption{}%A partial support nerve diagram for the Hardy model (\ref{fig:nerveHardy}), where the thick lines outline a quadruple-intersection or 3-simplex, illustrated in \ref{fig:nerve3}. The dashed line in \ref{fig:nerveHardy} is a 1-simplex that is not the edge of a 3-simplex, allowing for a violation of subadditivity: given any $z\in a_0\cap b_0$ such that $z\in Z$, there is some measure-zero set $E(s_U)\in\Sigma$, $U\in\M'$ such that $z\in E(s_U)$ (that is, $E(s_U)$ is not in the support).}
    \label{fig:nerves}
\end{figure}

In the case where $e$ is possibilistically contextual, there is also a beautiful topological interpretation of $e^\dag$'s violation of subadditivity that mirrors the bundle diagrams of Fig. \ref{fig:bundles}.
For instance, consider $e^\dag$ for the Hardy model. Fig. \ref{fig:nerveHardy} is a `nerve diagram' of the $\Sigma$-events in the support of $\mu$, meaning that $n$-simplices represent $n$-intersections of $\Sigma$-events. 
Notice that diagram is essentially the bundle diagram \ref{fig:Hardy}, except that the AB-event $a\mapsto 0$ has been replaced by the $\Sigma$-event (or $0$-simplex) $a_0:=E(a\mapsto 0)$, the AB-event $\{a\mapsto 0,b\mapsto 0 \}$ with the $\Sigma$-event (or $1$-simplex) $a_0\cap b_0$, and so on. 
The relevant violation of subadditivity can then be understood as stemming from the fact that the dashed $1$-simplex is not the edge of a $3$-simplex (the WPS-analogue of a global section) whose edges are in the support of $\mu$.\footnote{Clearly, this suggests a cohomological formulation of additivity-violation, although we have not the room to treat this aspect here.}

Two important corollaries immediately follow from Thm. \ref{thm:main}.
First, since AB (\citeyear{AB}) have shown that the Hardy model is logically contextual but not strongly contextual, and that the Bell model is probabilistically contextual but not logically contextual, our result shows that there are quantum models which `inhabit' lower tiers of the hierarchy of additivity-violation.
It also extends the original NFNC-violation result of Feintzeig and Fletcher (\citeyear{Feintzeig2017}) by showing that \textit{non-quantum} models like the PR Box have WPS representations which maximally violate subadditivity.

To state the second corollary, we will need a final notion of `extension' discussed in (\cite{Wronski2017, Feintzeig2014}): a \textit{classical extension} of $W$ is an extension $(Y,\Sigma',\mu')$ that is a probability space, i.e. $\Sigma'$ is an algebra, and $\mu'$ is additive.
The corollary is that $e$ is contextual (either strongly, or logically, or probabilistically) only if its WPS representations do not have a classical extension. Thm. 2 of (\cite{Wronski2017})---a minor adaptation of a result of Paris (\citeyear{Paris2001})---states that an abstract WPS is Dutch Bookable iff it does not have a classical extension. Thus, if $e$ is contextual then $e^{\dag}$ is Dutch Bookable.
The complete set of relations between our two hierarchies is shown in Fig. \ref{fig:hierarchy}.

% \begin{figure}[!htb]
% \begin{center}
% \scalebox{0.8}{
% \begin{tabular}{c c c}
%      \Centerstack[c]{$e$ is strongly\\ contextual} & $\Leftrightarrow$ & \Centerstack[c]{$e^{\dagger}$ maximally violates subadditivity}\\
%      $\Downarrow$ &  & $\Downarrow$\\
%      \Centerstack[c]{$e$ is logically\\ contextual} & $\Leftrightarrow$ & \Centerstack[c]{$e^{\dagger}$ violates  subadditivity}\\
%      $\Downarrow$ &  & $\Downarrow$\\
%      \Centerstack[c]{$e$ is probabilistically \\ contextual} & $\Leftrightarrow$ & \Centerstack[c]{Any monotonic extension of $e^{\dagger}$ violates additivity}\\
%      $\Updownarrow$ &  & $\Updownarrow$\\
%      \Centerstack[c]{$e$ has no NC,\\ factorizable HVT} &  & \Centerstack[c]{$e^\dag$ has no classical extension} \\
%       &  & $\Updownarrow$ \\
%       &  & \Centerstack[c]{$e^\dag$ is Dutch Bookable}
     
% \end{tabular}}
% \end{center}
% \caption{The hierarchy of contextuality (left) and the hierarchy of Dutch Bookability from additivity-violation (right).}
% \label{fig:hierarchy}
% \end{figure}

We note that the aforementioned result of Paris (\citeyear{Paris2001}) is essentially one about `convexity', which lies at the heart of the formal concept of a Dutch Book.
The next section reviews formal Dutch Books and describes the hierarchy of contextuality in terms of a hierarchy of convexity-violation. 

%While we give a detailed proof of Prop. 1 in the Appendix, it is possible to give the reader an intuition for why Prop. 1 is true.

%Note a trivial consequence of Claim 1: if a local section $s$ in the support of $e$ fails to extend to a global section in the support, then any $z\in\Ebar(s)$ must live in some measure zero set. To see this, note that either $z$ lives in a set in $D_1\cup D_2$, or it lives in the image of a global section under the event transfer map; in the latter case, there will be some other local section whose image under this map contains $z$ and is measure zero. Thus: if no local section has a global extension, every point in the space will live in a measure-zero set, and so FNC will be satisfied.

%%%%%%%%%%%%%%%%%%%%%%%%%%
% DB and convex analysis %
%%%%%%%%%%%%%%%%%%%%%%%%%%

\subsection{Dutch Bookability and convexity-violation}\label{sec:convexity}

The violation of the formal `No Dutch Books' constraint and its standard normative interpretation is well-trodden territory, so we shall be brief in our summary. 
Our main aim is to stress that since formal Dutch Bookability is a statement about convexity-violation, it follows that AB's hierarchy of contextuality can be straightforwardly understood as a hierarchy of convexity-violation (cf. Prop. \ref{prop:convexity} below). 
% \footnote{We highlight the links between convexity-violation in the WPS framework and AB's original qualitative hierarchy. We do not consider links to the `contextual fraction' explored by Abramsky, Barbosa, et al. (\citeyear{Abramsky2017b}), which quantitatively measures contextuality by decomposing an empirical model into a convex combination of a noncontextual empirical model and a strongly contextual one (the weight of the latter being the aforementioned `fraction').}
In order to separate this point from normative issues, we will defer our review of the standard normative interpretation of formal Dutch Books to Section \ref{sec:normative}.

\begin{figure}
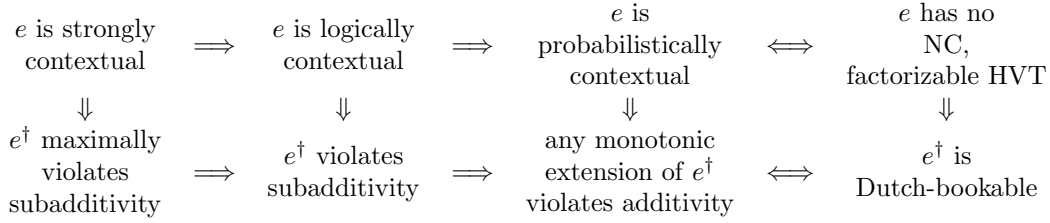

\begin{center}
\begin{tabular}{c c c c c c c}
     \Centerstack[c]{$e$ is strongly\\ contextual} & $\Longrightarrow$ & \Centerstack[c]{$e$ is logically\\ contextual} & $\Longrightarrow$ & \Centerstack[c]{$e$ is \\ probabilistically \\ contextual} & $\Longleftrightarrow$ &  \Centerstack[c]{$e$ has no \\ NC,\\ factorizable HVT}  \\
     $\Downarrow$ & & $\Downarrow$ & & $\Downarrow$ & & $\Downarrow$ \\
     \Centerstack[c]{$e^{\dagger}$ maximally\\ violates \\ subadditivity} & $\Longrightarrow$ & \Centerstack[c]{$e^{\dagger}$ violates \\ subadditivity} & $\Longrightarrow$ &  \Centerstack[c]{any monotonic \\ extension of $e^{\dagger}$ \\ violates additivity} & $\Longleftrightarrow$ & \Centerstack[c]{$e^{\dagger}$ is\\ Dutch-bookable}
\end{tabular}
\end{center}
\caption{The hierarchy of contextuality (top) and the hierarchy of Dutch Bookability from additivity-violation (bottom).}
\label{fig:hierarchy}
\end{figure}

% \vspace{\baselineskip}
% \noindent
% \emph{(i) The formal correspondence} 

Let $(Y, \Sigma, \mu)$ be a WPS. 
To succintly state the standard definition of a formal Dutch Book (adapted to WPSs), it will be useful to define the (finite) set of functions $\mathbb{V}:= \{ ~V_y = \chi_{(~\cdot~)} (y): \Sigma \rightarrow \{0, 1\}  ~\}$, where $y\in Y$ and $\chi_A: Y \rightarrow \{0,1\}$ is the characteristic function of $A \in \Sigma$. 
Each function $V_y \in \mathbb{V}$ can be identified with the `elementary/atomic event' specified by an element $y$ of the sample space $Y$. 
% we note that in general, such `events' will of course not be $\Sigma$-events. 
We say that a WPS violates a `No Dutch Books' constraint (or: is Dutch Bookable) iff there exists a function $s: \Sigma \rightarrow \mathbb{R}$ and a finite collection of sets $\V \subset \Sigma$ such that
\beq\label{eq:DB}
\text{for any $V_y \in \mathbb{V}$, } \sum_{A \in \V} s (A) \cdot (V_y (A)-\mu(A)) <0.
\eeq
\noindent
It is also common to leave $(Y, \Sigma)$ implicit and simply say that the set-function $\mu$ is Dutch Bookable. 

It is a well-known fact that condition (\ref{eq:DB}) is a statement about convexity-violation.
More precisely: let $\text{Conv}(\mathbb{V})$ be the convex hull of $\mathbb{V}$, i.e. the set of all convex combinations of functions in $\mathbb{V}$. We will say that $\mu$ is $\mathbb{V}$-convex iff $\mu\in\text{Conv}(\mathbb{V})$. Then, as originally noted by de Finetti (\citeyear{deFinetti-ToP:1974}), (for finite $\mathbb{V}$) $\mu$ is not Dutch Bookable iff $\mu$ is $\mathbb{V}$-convex.\footnote{See also Theorem 2 of Paris (\citeyear{Paris2001}).}

With this background in place, it is easy to see that AB's hierarchy of contextuality can be interpreted as a hierarchy of obstructions to $\mathbb{V}$-convexity. For $\V \subseteq \Sigma$, we say that a WPS representation $e^{\dagger}$ $\V$-violates $\mathbb{V}$-convexity so long as $\mu|_{\V} $ cannot be written as a convex sum of $V_y|_{\V}$ for $V_y\in \mathbb{V}$; this will constitute the lowest tier of the hierarchy. 
In order to describe the higher tiers, let us introduce the Boolean addition operator $\oplus$, i.e. $1\oplus 1 =1$, $0\oplus 1=1\oplus 0 = 1$, and $0\oplus 0 = 0$ (logical sums of functions are conducted pointwise). 
Let $\chi_\mu$ be the characteristic function of the support of $\mu$ (i.e. for $A\in\Sigma$, $\chi_\mu(A)=0$ if $\mu(A)=0$ and $\chi_\mu(A)=1$ otherwise). 
We will say that $e^{\dagger}$ logically $\V$-violates $\mathbb{V}$-convexity iff $\chi_\mu |_\V$ is not a logical sum of $V_y|_{\V}$ for $V_y\in \mathbb{V}$; clearly, in such a case, no convex combination of $V_y$ will yield $\mu$. In this case, however, it may still be that
\begin{equation}\label{eq:detfrag}
 \chi_\mu|_\V = f+ \lsum_{V_y\in \mathbb{V}'} V_y|_{\V}
\end{equation}
for some non-empty $\mathbb{V}'\subseteq \mathbb{V}$ and some function $f:\V\to\{0,1\}$. 
If (\ref{eq:detfrag}) cannot be satisfied for $\V$, then we will say that $e^{\dagger}$ strongly $\V$-violates $\mathbb{V}$-convexity. Given these definitions, one immediately obtains the following correspondence between contextuality and convexity-violation (we sketch a proof in the Appendix):
\begin{prop}\label{prop:convexity}For any empirical model $e$ and any of its combinatorial WPS representations $e^\dag$,

\begin{enumerate}
    \item $e$ is strongly contextual iff $e^{\dagger}$ strongly $\V_{\M}$-violates $\mathbb{V}$-convexity.
    \item $e$ is logically contextual iff $e^{\dagger}$ logically $\V_{\M}$-violates $\mathbb{V}$-convexity.
    \item $e$ is probabilistically contextual iff $e^{\dagger}$ $\V_{\M}$-violates $\mathbb{V}$-convexity.
\end{enumerate}
\end{prop}
\noindent
By de Finetti's result, the probability measure $\mu$ of $e^{\dagger}$ violates $\mathbb{V}$-convexity iff it is Dutch Bookable; thus, this hierarchy of violations of $\mathbb{V}$-convexity is a hierarchy of conditions witnessing Dutch Bookability.

\section{Is Dutch Bookability normative?}\label{sec:normative}

In this paper, we have demonstrated a correspondence between AB's hierarchy of contextuality and a hierarchy of Dutch Bookability for WPS representations. This shows that (Standard) arguments against CHVTs can be mapped to corresponding (Dutch) arguments for WPS representations. In this final section, we turn to a partial evaluation of whether formal Dutch Bookability can be used to provide a normative argument against GHVTs. 

Feintzeig and Fletcher (\citeyear{Feintzeig2017}) seek to use (Dutch) arguments to demonstrate that the probabilities of hidden variable theories based on WPSs---which we have labeled GHVTs---cannot be given a `rational' subjective intepretation. To unpack this claim, first recall that the existence of a `finite null cover' constitutes a maximal violation of subadditivity, the strongest rung of the Dutch Bookability hierarchy (pictured in Figure \ref{fig:hierarchy}). Now, note that Feintzeig and Fletcher's specific claim is that `whenever a weak probability space [has a finite null cover], anyone who sets their beliefs according to that probability space will be irrational in the sense that she will accept as fair a series of bets on which she is guaranteed to lose money' (\citeyear[p. 308]{Feintzeig2017}).

This is a dramatic claim! As Feintzeig and Fletcher note, WPSs are widely used in the literature,  and several authors explicitly use such structures as the basis for hidden variable theories (\cite{Feintzeig2017,Craig2006,Suppes1991}). Is a subjective approach to the probabilities in any such GHVT irrational, as Feintzeig and Fletcher seem to suggest?
We will not provide a complete assessment of this claim here; however, we wish to highlight that on one natural and plausible reconstruction, the soundness of Feintzeig and Fletcher's normative argument requires \emph{at least} a very specific approach to the `hidden variable' interpretation of a WPS.

To unpack this approach, it is helpful to consider what GHVTs (and their classical counterparts) \emph{could} be. As standardly construed in the literature, and by our lights, a hidden variable theory (e.g. a CHVT or a GHVT) should include at least the following schematic ingredients:
\begin{enumerate}[label=(\roman*)]
    \item A structure of \emph{events};
    \item \emph{Probabilities} of these events;
    \item \emph{States of affairs} and their associated \emph{response functions} (recall that a response function specifies the probability of an event conditional on the system's being in some particular state);
    \item A \emph{relationship} between the probabilities of events and the response functions.
\end{enumerate}
For instance, Abramsky and Brandenburger flesh out this schema for a CHVT as follows. Let $\mathcal{E}= \{ E \}$ be the set of events described in Section \ref{sec:gluing_formal} (the reader re-joining us from Section \ref{sec:con_in} may safely view $\mathcal{E}$ as a set without any additional structure)---this specifies (i). Let the Born rule specify the empirical probabilities $p(E \,|\, C)$, where $C\in\mathcal{M}$ is a maximal set of compatible measurements that can yield $E$---this specifies (ii). Now specify (iii) by introducing a (finite) set of states of affairs $\Lambda$. We take each $\lambda\in\Lambda$ to specify a response function $p (~\cdot~ \,|\, \lambda,C)$ for each $C\in\mathcal{M}$, which maps all events that can be measured in the specified context to weights in the unit interval. Finally, introduce functions $p(~\cdot~ ):\Lambda \to [0,1]$ such that $\sum_\lambda p(\lambda)= 1$, which describe the probabilities that the preparation of the system yields a given state of affairs. 
% \footnote{We have made, here, the standard assumption of `$\lambda$-independence'---i.e. the context of measurement does not affect $p(\lambda)$ (\cite{AB}).}
(iv) is then specified by requiring that the empirical probabilities are reproduced by means of the following functional relationship:
\begin{equation}\label{eq:spekkens}
    p (E\,|\,C) = \Sum_{\lambda\in \Lambda}  p(\lambda) p (E\,|\,\lambda,C).
\end{equation}
This approach allows for a transparent definition of `noncontextual', which is also a standard definition in the quantum foundations literature. The definition is that response functions are noncontextual just in case $p (E \,|\, \lambda,C) = p (E\,|\,\lambda,C')$ for all contexts $C,C'\in\mathcal{M}$ that can yield the measurement event $E$.\footnote{Note that this definition of `noncontextual' aligns with Spekkens's (\citeyear{Spekkens2005}) notion of `measurement noncontextuality' (albeit without his operational definition of `measurement contexts').}

% Note, too, that it aligns with the condition of `parameter independence' (\cite{AB,Leifer2014a}) when the measurements in question concern bipartite systems.

%We also recall from Section \ref{sec:con_in} that response functions are said to be outcome-deterministic (OD) just in case they only assign values of 0 or 1. Given this terminology, what (Standard) arguments demonstrate is that no NC, OD response functions can recover all probabilities $e_C(E)$ via the functional relationship (\ref{eq:spekkens}).\footnote{We refer the reader to (\cite{Kunjwal2015}) for a lucid discussion of the relationship between this statement and the equivalent statement for NC factorizable hidden variable theories.}

How might GHVTs satisfy the above schema? At the very least, all the WPS-users mentioned above use the set $\Sigma$ of WPS events to represent the structure of measurement events, and they use the WPS `measure' $\mu$ to represent the empirical probabilities given by the Born rule (the usual probabilistic rule of quantum theory). 
% \footnote{Each event $E$ will correspond to an effect on some Hilbert space, though we note that the structure of events in a hidden variable theory \emph{need not} inherit any of the usual structures defined on operator algebra, such as the lattice structure on projections or the poset structure on effects.}
This specifies (i) and (ii), and we take these choices to be uncontroversial. 

However, Feintzeig and Fletcher do not specify (iii) and (iv) (at least explicitly). We now provide a charitable reconstruction of how one might do so in order to construct a valid (Dutch) argument against the `coherence' of GHVTs.
%Charitably, Feintzeig and Fletcher do have something specific in mind with regards to (iii) and (iv)---but to see what this is, we must return to their interpretation of Dutch Books. We now briefly review 
Let us first recall the standard definition of Dutch Bookability generalized to WPSs (c.f. Section \ref{sec:convexity}) for the reader re-joining us from Section \ref{sec:con_in}. Let $Y$ be the sample space of a WPS, let $\Sigma$ (a subset of the power set of $Y$) be the set of WPS-events, and let $\mathcal{V}$ be some finite subset of $\Sigma$. Now, in the spirit of classical probability theory, define for each point $y$ in the sample space $Y$ an `elementary event' $V_y := \chi_{(~\cdot~)}(y) : \Sigma \to [0,1]$, where $\chi_{(A)}(y)$ is the characteristic function of the WPS-event $A$ evaluated at $y$. Now, we say a WPS set-function $\mu:\Sigma\to[0,1]$ violates a `No Dutch Books' constraint (or: is Dutch Bookable) just in case there is some set $\mathcal{V}$ and some function $s : \mathcal{V} \to \mathbb{R}$ specifying real-valued `stakes' such that
\beq\label{eq:DB2}
\text{for any $V_y \in \mathbb{V}$, } \sum_{A \in \V} s (A) \cdot (V_y (A)-\mu(A)) <0.
\eeq
\noindent
There is a standard `normative Dutch Book' story from formal epistemology whose aim is establish that Dutch Bookability is pathological (\cite{Hajek2008,deFinetti-ToP:1974}). This story interprets equation (\ref{eq:DB2}) as a betting game. Suppose a bookie proposes a game wherein each WPS event $A\in\mathcal{V}$ is valued at $s(A)$ dollars. If you agree to the game, then if $V_y$ describes the state of the world, you will receive (or be forced to pay the bookie, if $s(A)$ is negative) $s(A) \cdot V_y(A)$ dollars when the game ends. However, to play the game, you must put down (or agree to be paid by the bookie, if $s(A)$ is negative) $s(A) \cdot \mu(A)$ dollars. If $\mu(~\cdot~)$ truly reflects your expectations for the values of $V_y(~\cdot~)$ (the thought goes), then you should accept this game as `fair'. But then if equation (\ref{eq:DB2}) holds, you will accept as fair a series of bets on which you will lose money no matter what the world turns out to be like.

There are many reasonable philosophical worries about this story that we will not address here.\footnote{For some of these, see Kaplan (\citeyear{Kaplan1998}).} Instead, our concern is with the following point: in order for normative Dutch Books to make any sense for WPSs, one has to take seriously the idea that the elementary events of a WPS have something interesting to say about \emph{what the world is like}. What could this something be?
By our lights, the most reasonable answer is to take the elementary events of a WPS to be response functions, thereby specifying (iii). This move also gives a sense to Feintzeig and Fletcher's assertion that the relevant GHVTs are `non-contextual' (a term they do not define): these response functions are non-contextual in the sense defined above (just after (5.1)). 

The above specification of (iii) allows us to clearly lay out a certain subjective approach to specifying (iv), and to explain how a (Dutch) argument can be launched against it. The approach interprets the `empirical probabilities' $\mu$ as an agent's subjective expectations for the values of response functions. (iv) is then specified by an argument that the functions $\mu$ are all and only those expectations that satisfy a certain coherence condition, which is standardly taken to be the avoidance of Dutch Books (cf. de Finetti (\citeyear{deFinetti-ToP:1974})). In other words, a GHVT whose response functions are taken to be $V_y(~\cdot~)$ has a `coherent' $\mu$ just in case constraint (5.2) is satisfied. On this basis, we can reconstruct Feintzeig and Fletcher's original (under-specified) argument as saying that, given (i)--(iv), (Dutch) shows that the relevant GHVT is `incoherent'.  
%What our Dutch Bookability results (cf. Thm. 1 and Prop. 1) 
%They borrow this interpretive strategy from physical theories using classical probability spaces, as elementary events are often taken to represent (outcome-deterministic, noncontextual) response functions in those theories. They then use their (Dutch) argument to critique a certain subjective approach to (iv).  This approach interprets the `empirical probabilities' $\mu$ as subjective expectations for the values of the elementary events. One way to defend the `rationality' of this interpretation---spearheaded by de Finetti (\citeyear{deFinetti-ToP:1974})---is to show that the probabilities $\mu$ fix all and only those expectations satisfying a certain coherence condition. Now, Feintzeig and Fletcher's argument shows that---given our assumptions about (i)--(iii)---we cannot fulfill (iv) in this subjective way when we make the standard assumption that an expectation is `coherent' just in case it `avoids Dutch Books'.

We now emphasize two morals concerning this reconstruction of the normative (Dutch) argument against GHVTs. First, Feintzeig and Fletcher (\citeyear{Feintzeig2017}) motivate GHVTs by suggesting (\textit{\`{a} la} \cite{Pitowsky1989}) that they can escape one (Standard) argument---Bell's theorem---by weakening the strictures of classical probability theory. But our reconstruction requires specifying NC, OD response functions for GHVTs. Any HVT with such response functions is `factorizable'---very roughly, joint distributions for several outcomes in one context factor as the product of the probabilities of each outcome. And by Fine's theorem (\citeyear{Fine1982}), if quantum probabilities witness a violation of Bell inequalities, then there \emph{cannot be} such an HVT---so long as the quantum probabilities ought to be reproduced by equation (\ref{eq:spekkens}).\footnote{See Kunjwal (\citeyear{Kunjwal2015}) for a lucid discussion of this point in the ontological models framework.} It was noted in Section \ref{sec:gluing_formal} that the non-existence of such an HVT is an example of the third and lowest tier of AB's hierarchy of contextuality. The more general moral, then, is that GHVTs with NC, OD response functions are susceptible to (Standard)---and \textit{thereby} (Dutch)---arguments; thus, (Dutch) arguments simply re-express (Standard) arguments within the language of the WPS framework.

Second, while Feintzeig and Fletcher seem to suggest that a subjective approach to the probabilities of \emph{any} GHVT (representing a Kochen-Specker model) is irrational because it is Dutch Bookable, this consequence can be avoided by any WPS-user who divests elementary events of physical meaning and specifies response functions in some other way. For example: if response functions are given by wavefunctions in a finite-dimensional Hilbert space, then all and only the probabilities given by density matrices avoid Dutch Books.\footnote{See Steeger (\citeyear{steeger/online}) for a demonstration of this point, as well as a generalization of it, in the $C^*$-algebraic framework.}

We conclude by noting that once enough of the mathematical features of a WPS (such as `elementary events') are stripped of physical meaning, one may well begin to question the usefulness and interest of WPSs as the backbone of a hidden variable theory. %\footnote{This is not to say that hidden variable theorists should cease investigating modifications of classical probability theory---far from it! We claim only that WPSs may not be the right tools for their job.} 
Our own view is that WPSs are primarily of interest because they are a quasi-probabilistic framework within which we can represent quantum (and non-quantum) models and classify different strengths of contextuality, analogous to Abramsky and Brandenburger's sheaf-theoretic framework---this is the sense in which our results about the hierarchy of Dutch Bookability are of foundational significance.
But such structures can be interesting while still failing to be useful for a hidden variables approach to quantum theory.

\section*{Acknowledgements}

We thank Samson Abramsky, Jeremy Butterfield, Bryan Roberts, Alex Pruss, and Alan H\'{a}jek for helpful discussions of this material. Nicholas Teh thanks audiences at the Rutgers Philosophy and Probability colloquium series. Jeremy Steeger thanks audiences in the Controlled Quantum Dynamics group at Imperial College London. Our work on this project was supported by NSF Grant \#1734155.

%%%%%%%%%%%%%%%%%%%%%%%%%%%%%%%%%%%%%%%%%%%%%%%%%%%%%%%%%%%%%%%%%%
\renewcommand{\thesection}{A}
\section{Appendix}

Before presenting the proofs of our results, we include the following comment on the relationship between Feintzeig and Fletcher's \emph{weak hidden variable representations} and our \emph{WPS representations}. First, recall that Feintzeig and Fletcher directly construct representations of `quantum mechanical experiments', viz. triples $(\mathcal{H},\psi,\mathcal{O}_n)$ where $\mathcal{H}$ is a Hilbert space, $\psi\in \mathcal{H}$ is a normalized vector, and $\mathcal{O}_n = \{P_1, \mathellipsis, P_n \}$ is a finite set of projectors on $\mathcal{H}$. We also note that in (\cite{AB}), an empirical model $e$ is said to have a \emph{quantum representation} if there is some $\mathcal{O}_n = X$ whose maximal subsets of commuting elements form the elements of $\M$, $O=\{0,1\}$ (the spectrum of each $P$), and $e_C(s) = \langle \psi, P_s \psi \rangle$, where $P_s$ is the projector on $\mathcal{H}$ corresponding to the section $s\in \E(C)$. Thus, every quantum mechanical experiment is the quantum representation of the appropriate empirical model $e$.

In (\cite{Feintzeig2017}), a \emph{weak hidden variable representation} of a quantum mechanical experiment $(\mathcal{H},\psi,\mathcal{O}_n)$ is a weak probability space $(Y,\Sigma,\mu)$ and a map $E':\mathcal{O}_n \to \Sigma $ satisfying the following two conditions: $\mu (E'(P))= \langle \psi, P \psi \rangle $ for all $P\in \mathcal{O}_n$, and for all orthogonal $P_i,P_j \in\mathcal{O}_n $, $E'(P_i) \cap E'(P_j ) \in \Sigma $ and $\mu (E'(P_i) \cap E'(P_j )) = 0$. In order to derive a connection between these representations and additivity-violations, Feintzeig and Fletcher require that these spaces satisfy an analog of (WC): for every set $Q \subseteq \mathcal{O}_n$ of mutually orthogonal projectors spanning $\mathcal{H}$, and letting $\Sigma_Q$ be the algebra formed by $Q$, $\Sigma_Q \subseteq \Sigma$ and $(Y,\Sigma_Q,\mu |_Q)$ is a probability space.

Now let $e$ be the empirical model induced by $(\mathcal{H},\psi,\mathcal{O}_n)$. It is clear that by (WC), (EC), and (ME), any WPS representation $e^\dag$ will also be a weak hidden variable representation of $(\mathcal{H},\psi,\mathcal{O}_n)$ where $E'(P) = E( P \mapsto 1 )$ for all $P\in\mathcal{O}_n$. However, not every weak hidden variable representation of $(\mathcal{H},\psi,\mathcal{O}_n)$ will be a WPS representation of $e$; the former are blind to the structure of the event sheaf and so need not satisfy condition (\ref{eq:setglue}). We note that condition (\ref{eq:setglue}) imposes the minimal additional structure needed in order to derive interesting connections between sheaf-theoretic inconsistency and violations of additivity; these are established by Theorem \ref{thm:main}.

We now present the proofs of this theorem and our two propositions. We begin by reviewing some notation.
First, for AB-event sections, we will use a subscript to keep track of the set $U \subseteq X$ over which a section lies, i.e. $s_U \in \E(U)$ (where $X$ is the measurement set of an empirical scenario). For instance, $s_X$ will denote a particular global section.
Second, we will use $S_U := \Ebar (s_U)$ to simplify our notation for representations of AB-events (as sets within the WPS sample space $Y$); we will also use $\E^{\dag} (U) := \{ \Ebar (s_U ) : s_U \in \E (U) \}$ to denote the set of representations of AB-events over $U$.
Third, in what follows, we will always use $C$ and $C'$ to refer to \textit{maximal} contexts, i.e. $C,C' \in \M$.

In order to relate the different tiers of contextuality to properties of our WPS representations, we will want to excise the following two collections of sets from $Y$:
\beq
    D_1 := \left \{ E(s) \cap E(s')  :   s, s' \in \E(x) \text{ for some } x\in X \right \}.
\eeq

\beq
    D_2 := \left \{Y - \bigcup_{s\in \E(x) } E(s) :  x\in X \right \}.
\eeq
$D_1$ represents `contradictory events', whereas $D_2$ represents `measurements with non-existent outcomes'; neither of these types of `events' is represented within the AB framework, and by (WC),(EC) and (ME), every set in these collections is a measure-zero $\Sigma$-event.
Thus, it is the result of the excision $Z := Y - \bigcup(D_1\cup D_2)$ which will contain the properties of interest to us.
In the course of our analysis, we will frequently have reason to consider the intersection of $Z$ with some set $B \in P(Y)$; we thus introduce the notation:
\beq
\tilde{B} := Z \cap B
\eeq

Let $e$ be an empirical model, and let $e^\dag$ be any one of its WPS reprsentations.
We will now prove a key lemma that will be used repeatedly to help translate the `contextuality' of $e$ into the properties of $e^\dag$. 
\begin{lemma}\label{lem:notnull}
If $z\in Z := Y - \bigcup(D_1\cup D_2)$, then $z\in \Ebar(s_X)$ for some $s_X \in \E(X)$.
\end{lemma}

%We are now in a position to prove Lemma 4.2, viz. the statement that given a WPS representation $e^\dag$ of an empirical model $e$, any $z \in Z$ is contained in the set $S_X$ corresponding to some global section $s_X$.  
\begin{proof}[\unskip\nopunct]
\textit{Proof.} 

We begin by arguing that for all $x \in X$, any $z \in Z$ is contained in exactly one $S_x \in \E^{\dag} (x)$. 
Let us define the union 
\[
A_x := \bigcup_{S_x \in \E^{\dag}(x)} S_x
\]
for an arbitrary $x \in X$.
Since $D_2$ contains $A_{x}^{c}$ and $D_2$ is disjoint from $Z$ (by the definition of $Z$), it is clear that $Z \cap A_{x}^{c} =\emptyset$, and so $A_x \supseteq Z$.
In other words, any $z \in Z$ is contained in the union $A_x$.
To further see that any $z \in Z$ is contained in exactly one $S_x$ out of this union, we simply note that since $D_1$ is disjoint from $Z$, 
\beq
\tilde{S}_x \cap \tilde{S'}_x = \emptyset ~~ \text{for any}~S,S' \in \E^{\dag}(x),
\eeq
where we recall that $\tilde{S}:= S \cap Z$.

We now know that for any $z\in Z$, there is a collection $\{ S_x \}_{x \in X}$ such that each $S_x$ contains $z$.
Since any collection of AB-events $\{s_x \}_{x \in X}$ defines a global section $s_X$ (by the sheaf property of $\E$), and the corresponding $S_X$ is in the intersection of $\{S_x \}$ (by Eq. (\ref{eq:setglue})), we conclude that $z \in Z$ is contained in $S_X$.

%To complete the proof of the Lemma, we note that any collection of AB-events ${s_x}_{x \in X}$ defines a global %section $s_X$, and so by *** the corresponding set $S_X$ 

%Clearly, $A_x$ contains $Z$, because 
%It is easy to see that $A_x \supseteq Z$: simply observe that $ \tilde{A_{x}^{c}} = \emptyset$

%Recall all sets in $D_1$ and $D_2$ are measure zero by (Weak Classicality), and that $Z=Y-\bigcup(D_1\cup D_2)$. Note that for each $x\in X$, $\bigcup_{s\in \mathcal{E}(x) } E(s) \supseteq Z$. To see this, note that
%$$
 %   Z\cap \left ( Y-\bigcup_{s\in \mathcal{E}(x) } E(s) \right )=Z-Z\cap \left ( \bigcup_{s\in \mathcal{E}(x) } %E(s) \right ) =\emptyset,
%$$
%and so $Z=Z\cap \left ( \bigcup_{s\in \mathcal{E}(x) } E(s) \right )$. Note also that
%$$
    %(E(s)\cap Z) \cap (E(s')\cap Z) =\emptyset
%$$
%for any distinct $s,s' \in\mathcal{E}(x)$, which follows trivially from $D_1$ disjoint from $Z$. From these two %facts it immediately follows that for every $z\in Z$, for every $x\in X$, $z$ lives in exactly one element of the %set $\{ E(s) : s\in\mathcal{E}(x) \}$. We thus have $z\in \Ebar(s_X)$ for some $s_X\in \E(X)$.
\end{proof}

%\begin{proof}[\unskip\nopunct]
%\textit{Proof of Claim 2.} Follows directly from Wro\'{n}ski and Godziszewski's Theorem 2 [\citeyear{Wronski2017}]---note the $\Sigma$ of a nonclassical WPS, just like that of a nonclassical generalized probability space, will be a proper subset of a $\sigma$-algebra on $Y$.
%\end{proof}

We now use this lemma to prove Theorem 1. 

% We also remind the reader that in the statement of the theorem, we only consider extensions of $e^\dag$ such that if $e^\dag$ satifies subadditivity, then its extensions also satisfy subadditivity (this is because we only want to consider extensions which do not add any additivity-violation beyond what is already in $e^\dag$).

\begin{proof}[\unskip\nopunct]
\textit{Proof of Theorem 1.}

\vspace{\baselineskip}
\noindent
\emph{(1) $e$ strongly contextual $\Longrightarrow$ $e^\dag$ maximally violates subadditivity}

\vspace{\baselineskip}
In order to show the existence of a collection of sets that maximally violate subadditivity, it will be convenient to define $D_3$, the collection of measure-zero sets $S_C$ each of which contains some representation $S_X$ of a global section. More formally:
\[
D_3 := \{ S_C ~\text{corresponding to some}~ s_C : \mu(S_C)=0 ~\text{and}~ S_C \supset S_X ~\text{corresponding to some}~ s_X \}.
\]
Our strategy will be to use strong contextuality to show that the union of sets in $D_3$ contains $Z$, thus providing---in conjunction with $D_1$ and $D_2$---the collection of sets which maximally violates subadditivity.

We now translate `strong contextuality' into a feature of the collection $D_3$.
Recall that strong contextuality means that there are no global sections in the support of $e$; in other words, for any global section $s_X$, there exists a section $s_C:=s_{X}|_{C}$ such that $e_C (s_C)= 0$. 
Thus, by (EC) and condition (\ref{eq:setglue}), strong contextuality implies that any $S_X$ is contained in some measure-zero $S_C$; hence such an $S_C$ belongs in $D_3$. 
We can now combine this fact with Lemma \ref{lem:notnull} to deduce that---since any $z \in Z$ is contained in some $S_X$---strong contextuality implies that $\bigcup D_3 \supset Z$.

%But as we have just established, each such $S_X$ is in turn contained in an element of $D_3$, and so $\bigcup D_3 \supset Z$. We can now use this fact to produce the collection of sets that maximally violates subadditivity.
In light of this last WPS characterization of strong contextuality, we will take our violation-inducing collection to be $\V_3:= D_1 \cup D_2 \cup D_3$.
Since $Y = \bigcup \V_3$ and $\mu(Y)=1$ by (WC), we can immediately compute that $\mathfrak{a}(\V_3) = \mu ( \bigcup V_3 ) - \Sum_{A \in \V_3} \mu(A) = 1$, i.e. $\V_3$ yields a maximal violation of subadditivity. 

\vspace{\baselineskip}
\noindent
\emph{(2) $e$ logically contextual $\Longrightarrow$ $e^\dag$ violates subadditivity}

\vspace{\baselineskip}
Our strategy here is similar, viz. we will find a collection of sets that covers $Y$ and use this to show that the collection violates subadditivity.
However, our previous argument will need to be supplemented, because logical contextuality only guarantees the existence of \textit{one} section $s^* \in \E(C)$ in the support of $e$ which does not extend to a global section in the support of $e$.
Thus, we will only be able to show that the union of sets in $D_3$ contains a subset of $Z$; in order to construct a collection that completely contains $Z$, we will need to introduce a new collection that we call $D_4$. 

From the definition of logical contextuality, we know that for all global sections $s_X$ that restrict to $s^*$, there exists a distinct maximal context $C'$ such that the restriction $s_X|_{C'}=s_{C'}$ is not in the support of $e_{C'}$, i.e. $e_{C'} (s_{C'})=0$.
Thus, by (Empirical Consistency) and condition (\ref{eq:setglue}), $S_{C'}$ is a measure-zero set that contains any such $S_X$, i.e. it is an element of $D_3$.

Combining this fact with (\ref{eq:setglue}) and Lemma \ref{lem:notnull} now lets us argue that $\bigcup D_3$ contains $\tilde{S^*}:= S^* \cap Z$. 
To see this, we first note that by (\ref{eq:setglue}) and Lemma \ref{lem:notnull}, $S^* = \bigcup_{S_X : S_X|^{C}=S^*} S_X$.
And since we have just argued that any $S_X$ that extends to $S^*$ is contained in $S_{C'}$, it follows that $\bigcup D_3 \supset \tilde{S^*}$.

However, to complete our argument, we will need to generate a collection whose union contains all of $Z$. 
In order to do so, we first define another collection $D_4 := \{ S_C : S_C \neq S^* \}$ (i.e. sets that represent sections over $C$ which are not $s^*$), and note that by (EC), any set in $D_4$ must have measure less than one because---by hypothesis---$\mu (S^*) \neq 0$.
Supplementing $D_3$ with $D_4$ then allows us to construct the desired collection: $\bigcup (D_3 \cup D_4) \supset Z$, because by Lemma \ref{lem:notnull} and (\ref{eq:setglue}), any $z \in Z$ is contained in precisely one $S_C$.

We can now produce our subadditivity collection by defining $\V_4:= D_1 \cup D_2 \cup D_3 \cup D_4$, and noting that $Y= \bigcup \V_4$.
Since the sets in $D_1, D_2, D_3$ are all of measure zero and the sets in $D_4$ are of measure less than one, it follows that
$\mathfrak{a}(\V_4) = \mu ( \bigcup V_4 ) - \Sum_{A \in \V_3} \mu(A) > 0$.
Thus, $\V_4$ violates subadditivity. 
\\\\
\noindent
\emph{(3) $e$ probabilistically contextual $\Longrightarrow$ any monotonic extension of $e^\dag$ violates additivity}
\\\\
Note that if $e$ is possibilistically (i.e. either strongly or logically) contextual, then any monotonic extension of $e^\dag$ will violate subadditivity, and so it will also violate additivity. 
Thus, we will only need to consider the case in which $e$ is probabilistically, but not possibilistically, contextual.
In this case, $e^\dag$ satisfies subadditivity, and its monotonic extensions which are \emph{not} subadditive are not additive by default; thus, we need only consider the monotonic extensions of $e^\dag$ which satisfy subadditivity.
As such, we will now argue that if $e$ is probabilistically contextual, then any subadditive and monotonic extension of $e^\dag$ will contain a collection of sets $\V$ which violates additivity.

By probabilistic contextuality and (EC), there must exist some $S^* \in \E^\dag (C)$ such that $\mu(S^*)$ cannot be recovered by marginalization from any $\mu'$, i.e.
\beq\label{nomarg}
\mu (S^*) \neq \Sum_{\{S_X \in \E^{\dag}(C) :~ S_X|^{C} = S^*\}} \mu' (S_X).
\eeq

The non-marginalization equation (\ref{nomarg}) bears a resemblance to the additivity-violation condition; however, there is insufficient information to show that it is defined on a disjoint collection $\{S_X \}$ whose union yields $S^*$.
In order to obtain a related collection of sets that satisfies these two properties, we define $\V:= \{ \tilde{S}_X : S_{X}|^{C} = S^* \}$. 
We then observe that by Lemma \ref{lem:notnull}, this is a disjoint collection; and that by Lemma \ref{lem:notnull} and (\ref{eq:setglue}), $\bigcup \V = \tilde{S^*}$, which relates the collection to the LHS of (\ref{nomarg}).

In order to relate the measures of the sets used in (\ref{nomarg}) to  the measures of the sets in $\V$, we will need to establish two further facts, viz. $\mu' ( \tilde{S}_X)= \mu'(S_X)$ and $\mu' ( \bigcup \V) = \mu' (\tilde{S}^*) = \mu (S^*)$.
We do so by arguing that $\mu' ( \tilde{A})= \mu' (A)$ for any $A \in \Sigma'$. 
Let $\bar{A}:= A \cap Z^c $. To obtain the result, note that $\mu' (Z^c)=0$ by the subadditivity of $\mu'$ and so $\mu' (\bar{A})=0$ by the monotonicity of $\mu'$.
Then observe that by subadditivity, $\mu' (A) = \mu' ( \tilde{A} \cup \bar{A} ) \leq \mu' (\tilde{A})) + \mu' (\bar{A}) \leq \mu'( \tilde{A})$, and again by monotonicity, $\mu' (\tilde{A}) \leq \mu'(A)$.
Thus the result follows.

By means of the above facts, we can now use (\ref{nomarg}) to compute that $\mu' ( \bigcup \V) \neq \Sum_{A \in \V} \mu' (A)$, thus showing that $\V$ witnesses an additivity-violation for  $(Y, \Sigma', \mu')$.

\end{proof}

In preparation for the proofs of Propositions 1 and 2, we provide an illustration of a simple combinatorial WPS representation. Let $e$ be the empirical model of the Specker triangle, where $X=\{a, b, c\}$, $O=\{0,1\}$, and $\M = \{\{a,b\},\{b,c\},\{a,c\}\}$.\footnote{This is all the data we will need for our example, but see (\cite{Spekkens2011}) for the complete model.} A simple combinatorial WPS representation $e^{\dag}$ is obtained by letting the sample space $Y$ contain a unique point for each global section $s_X\in\E(X)$ and no additional points.\footnote{Note that the recipe for this simple WPS representation recovers the formalism of `experimental behaviors' used by (\cite{Dowker2014}).} This WPS can be graphically represented by the cube in Figure \ref{fig:combo}. In this figure, the points exhaust $Y$, edges represent events over maximal contexts, and faces represent atomic events. For example, the shaded face is $\Ebar(c\mapsto 0)=\{a_0b_0c_0,a_0b_1c_0,a_1b_0c_0,a_1b_1c_0\}$, and the thick edges represent the set $\{S : S \in \E^{\dag} (\{a,b\})\}$.\footnote{Feintzeig and Fletcher (\citeyear{Feintzeig2017}) discuss a variety of types of weak probability spaces that satisfy certain additional conditions, but not all of these will admit of such combinatorial models. The model illustrated, for example, fails to be a Generalized Probability Space (GPS) representation of the Specker triangle because there are disjoint $\Sigma$-events whose union is not a $\Sigma$-event.}

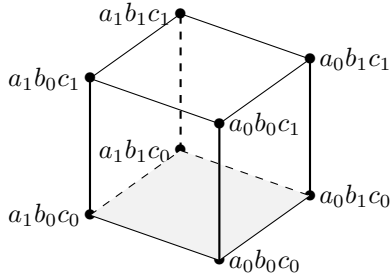
\begin{figure}[!htb]
    \centering
    %%%%%%%%%%%%%%%
    % Hardy model %
    %%%%%%%%%%%%%%%
        \tdplotsetmaincoords{60}{125}
        \begin{tikzpicture}[tdplot_main_coords,scale=2.1]
        \pgfmathsetmacro{\factor}{1};
        \pgfmathsetmacro{\inter}{1.5};
        
        \coordinate [label=left:$a_1b_1c_1$] (a1b1c1) at (0*\factor,0*\factor,1*\factor);
        \coordinate [label=right:$a_0b_1c_1$] (a0b1c1) at (0*\factor,1*\factor,1*\factor);
        \coordinate [label=left:$a_1b_0c_1$] (a1b0c1) at (1*\factor,0*\factor,1*\factor);
        \coordinate [label=right:$a_0b_0c_1$] (a0b0c1) at (1*\factor,1*\factor,1*\factor);
        
        \coordinate [label=left:$a_1b_1c_0$] (a1b1c0) at (0*\factor,0*\factor,0);
        \coordinate [label=right:$a_0b_1c_0 $] (a0b1c0) at (0*\factor,1*\factor,0);
        \coordinate [label=left:$a_1b_0c_0 $] (a1b0c0) at (1*\factor,0*\factor,0);
        \coordinate [label=right:$a_0b_0c_0$] (a0b0c0) at (1*\factor,1*\factor,0);
        
        \foreach \Point in {(a0b0c0),(a0b0c1),(a0b1c0),(a0b1c1),(a1b0c0),(a1b0c1),(a1b1c0),(a1b1c1)}{
            \node at \Point {\textbullet};
        }
        
        % ab sections
        % \draw[dashed, opacity=1, name path = a0b0] (a0)--(b0);
        % \draw[-, opacity=1, name path = a0b1] (a0)--(b1);
        % \draw[-, opacity=1, name path = a1b0] (a1)--(b0);
        % \draw[thick, opacity=1, name path = a1b1] (a1)--(b1);
        
        % ab' sections
        % \draw[-, opacity=1] (a0)--(b'0);
        \fill[gray!10]
        (a1b0c0)--(a1b1c0)--(a0b1c0)--(a0b0c0)--cycle;
        
        \draw[opacity=1] (a1b0c0)--(a0b0c0)--(a0b0c1)--(a1b0c1)--cycle;
        \draw[opacity=1] (a0b0c1)--(a0b1c1)--(a0b1c0)--(a0b0c0);
        \draw[opacity=1] (a1b0c1)--(a1b1c1);
        \draw[opacity=1] (a1b1c1)--(a0b1c1);
        \draw[dashed, opacity=1] (a1b0c0)--(a1b1c0);
        \draw[dashed, opacity=1] (a1b1c0)--(a0b1c0);
        \draw[dashed, thick, opacity=1] (a1b1c0)--(a1b1c1);
        
        \draw[thick] (a1b0c0)--(a1b0c1);
        \draw[thick] (a0b0c0)--(a0b0c1);
        \draw[thick] (a0b1c0)--(a0b1c1);
        \end{tikzpicture}
      \caption{A combinatorial WPS-representation of the Specker triangle.}
   
    \label{fig:combo}
\end{figure}

Before proving Proposition 1, we recall that $\V_{\M} := \{S_C : S_C \in \E^\dag(C), C\in \M  \}$ and that an `additive cover' is collection of disjoint sets $A \subset \Sigma$ that covers the sample space such that $\mathfrak{a}(A)=0$.

\begin{proof}[\unskip\nopunct]
\textit{Proof of Proposition 1.}

\vspace{\baselineskip}
\noindent
\emph{(1) $e$ strongly contextual $\Longleftrightarrow$ $e^\dag$ strongly $\V_{\M}$-violates subadditivity}

\vspace{\baselineskip}
Recall that $e^\dag$ strongly $\V_{\M}$-violates subadditivity iff the collection $V$ of measure-zero sets in $\V_{\M}$ covers \emph{all} the measure-nonzero elements of an additive cover $A\subseteq \V$. By (Exhaustiveness), we can set $A=\{S_C : S_C \in\E^\dag(C)\}$ for some $C\in\M$. It is clear that $e^\dag$ strongly $\V_{\M}$-violates subadditivity iff $\bigcup V = Y$. We will now show that this holds iff $e$ is strongly contextual.

Recall that strong contextuality means that there are no global sections in the support of $e$; in other words, for any global section $s_X$, there exists a section $s_C:=s_{X}|_{C}$ such that $e_C (s_C)= 0$.

Thus, by (EC), (\ref{eq:setglue}), and (Strong Mutual Exclusivity), every $S_X$ is contained in some measure-zero $S_C$ iff $e$ is strongly contextual. By (Exhaustiveness), $\bigcup_{S_X \in\E^\dag (X) } S_X = Y$, and by condition (\ref{eq:setglue}), no $S_X$ is trivial. Thus, $\bigcup V = Y$ iff $e$ is strongly contextual.

\vspace{\baselineskip}
\noindent
\emph{(2) $e$ logically contextual $\Longleftrightarrow$ $e^\dag$ logically $\V_{\M}$-violates subadditivity}

\vspace{\baselineskip}

Recall that $e^\dag$ \textit{logically $\V$-violates subadditivity} iff the collection $V$ of measure-zero sets in $\V\subseteq \Sigma$ covers \emph{some} non-measure-zero element $a$ of an additive cover $A\subseteq \V$ (and so $\mathfrak{a}(V\cup A\setminus \{a\} ) >0$).

Suppose $e$ is logically contextual; this means there is some section $s^* \in \E(C)$ for some $C\in \M$ which is in the support of $e$ but which does not extend to a global section in the support of $e$. By condition (\ref{eq:setglue}), for each global section $s_X$ that restricts to $s^*$, the associated WPS-event $S_X$ is contained in some $S_C \in V$. Again by (\ref{eq:setglue}), the union of all these $S_X$ is equal to $E(s^*)=S^*$; thus $\bigcup V$ contains $S^*$. Let $A= \{S_C : S_C \in \E^\dag(C) \} $ and note that $A$ covers $Y$ (by Exhaustiveness) and contains $S^*$. Thus, $(A\setminus \{S^*\})\cup V$ yields the desired logical violation of subadditivity, viz. $\mathfrak{a}(V\cup A\setminus \{S^*\} ) > 0$.

Now suppose $e$ is not logically contextual; this means \emph{all} sections $s_C$ in the support of $e$ extend to global sections $s_X$ in the support of $e$. By (\ref{eq:setglue}), any cover of a measure-nonzero $S_C$ must cover all of the global sections to which it extends; that is, it must cover all $S_X$ such that $S_X|^C=S_C$. Since at least one of these $S_X$ is in the support of $e^\dag$, by (EC) and (Strong Mutual Exclusivity), any cover using elements of $\V_{\M}$ must contain some measure-nonzero element. Thus no logical violation of subadditivity can be constructed.

\vspace{\baselineskip}
\noindent
\emph{(3) $e$ probabilistically contextual $\Longleftrightarrow$ $e^\dag$ $\V_{\M}$-violates additivity}
\vspace{\baselineskip}

Recall $e^\dag$ $\V_{\M}$-violates additivity iff the algebra generated by $\V_{\M}$ contains some collection $V$ of disjoint sets that violates additivity in any monotonic extension of $e^\dag$ (and so $\mathfrak{a}(V)\neq 0$ in these extensions).

First, note that $\E^\dag (X)$ is a collection of disjoint sets that generates the algebra $\Sigma'$ that is generated by $\V_{\M}$. Now note that if $e$ is \emph{not} probabilistically contextual, then (by definition) there exists some $e_X\in\D\E(X)$. This $e_X$ induces a monotonic extension $(Y,\Sigma',\mu')$ such that, by (EC) and (WC), all elements in $\Sigma'$ marginalize appropriately. That is: for any $S \in U$ for any $U\subseteq X$,
\beq
\mu (S) = \Sum_{\{S_X \in \E^{\dag}(X) :~ S_X|^{U} = S\}} \mu' (S_X).
\eeq
Since the sets in $\E^\dag (X)$ are all disjoint and generate the algebra, there is no collection of disjoint sets in this algebra that yield a violation of additivity. So there is an additive, monotonic extension of $e^{\dag}$.

Now suppose that $e$ is probabilistically contextual. By (EC), in any monotonic extension $(Y,\Sigma',\mu')$, there must exist some $S^* \in \E^\dag (C)$ such that $\mu(S^*)$ cannot be recovered by marginalization from $\mu'$. Let $V= \{S_X \in \E^{\dag}(X) :~ S_X|^{C} = S^*\}$, and note that $V$ is a collection of disjoint sets such that $\mathfrak{a}(V) \neq 0$.\footnote{These considerations are sufficient to demonstrate the stronger result that $e$ is probabilistically contextual if and only if \emph{any} (monotonic or non-monotonic) extension of $e^\dag$ contains a subset in $\Sigma'$ that violates additivity. We restrict our focus to monotonic extensions to emphasize the connection to Theorem \ref{thm:main}.}
\end{proof}

We now turn to Proposition 2. Its proof makes use of the following consequence of condition (\ref{eq:setglue}), (Strong Mutual Exclusivity), and (Exhaustiveness): letting $\mathbb{V}|_{\V_{\M}}:= \{V|_{\V_{\M}} : V\in \mathbb{V}\} $, there is a bijection
\beq
v:\E(X) \to \mathbb{V}|_{\V_{\M}} :: s_X \mapsto V^{s_X},
\eeq
where for $C\in \M$, $V^{s_X}(S_C)$ is 1 if $ S_X|^C = S_C$ and 0 otherwise. Additionally, if we let $\D_A\E(X)$ be the presheaf containing \emph{all} the distributions on X, then there is a bijection
\beq
d:\D_A\E(X) \to \text{Conv}(\mathbb{V}|_{S_{\M}}) :: e_X \mapsto \sum_{s_X \in \E(X)} e_X(s_X)\cdot v(s_X).
\eeq
By means of these two maps, we can demonstrate the equivalence of the hierarchy of contextuality and the hierarchy of convexity-violation.

\begin{proof}[\unskip\nopunct]
\textit{Proof of Proposition 2.}

\vspace{\baselineskip}
\noindent \emph{(1) $e$ strongly contextual $\ \Longleftrightarrow \ $ $e^{\dagger}$ strongly $\V_{\M}$-violates $\mathbb{V}$-convexity}

\vspace{\baselineskip}
$(\Longrightarrow)$ Suppose towards a contradiction that (\ref{eq:detfrag}) is satisfied for some $\mathbb{V}'$ and some $f$. Pick a $V\in\mathbb{V}'$; then $v^{-1}(V|_{\V_{\M}})=s_X$ must be in the support of $e$, which yields a contradiction.

$(\Longleftarrow)$ Suppose towards a contradiction that $s_X\in \E(X)$ is in the support of $e$. 
But then the set of $V\in\mathbb{V}$ which restricts to $v(s_X)$ satisfies (\ref{eq:detfrag}) for an appropriate choice of $f$, which yields a contradiction.

\vspace{\baselineskip}
\noindent \emph{(2) $e$ logically contextual $\ \Longleftrightarrow \ $ $e^{\dagger}$ logically $\V_{\M}$-violates $\mathbb{V}$-convexity}

\vspace{\baselineskip}
$(\Longrightarrow)$ Suppose that $\chi_\mu|_{\V_{\M}}$ is the logical sum of the elements of $\mathbb{V}'\subseteq \mathbb{V}$. Then every local section in the support of $e$ extends to a global section in the pre-image of $v$ over $\mathbb{V}'$, which yields a contradiction.

$(\Longleftarrow)$ Suppose that all local sections in the support of $e$ extend to global sections; then the logical sum of the image of $v$ over these global sections yields $\chi_\mu|_{\V_{\M}}$, which yields a contradiction.

\vspace{\baselineskip}
\noindent \emph{(3) $e$ probabilistically contextual $\ \Longleftrightarrow \ $ $e^{\dagger}$ $\V_{\M}$-violates $\mathbb{V}$-convexity}

\vspace{\baselineskip}
$(\Longrightarrow)$ Suppose otherwise, viz. suppose that some convex combination of $\mathbb{V}$ yields $\mu|_{\V_{\M}}$. It follows that $d^{-1}(\mu|_{\V_{\M}})$ is a distribution $e_X$ that marginalizes to each $e_C$, which yields a contradiction.

$(\Longleftarrow)$ Suppose otherwise, viz. suppose that there is some $e_X \in \D\E(X)$. Then $d(e_X)$ yields an element of $\text{Conv}(\mathbb{V})$ which is equal to $\mu|_{\V_{\M}}$; contradiction.

\end{proof}

%%%%%%%%%%%%%%%%%%%%%%%%%%%%%%%%%%%%%%%%%%%%%%%%%%%%%%%%%%%%%
% \bibliography{contextuality.bib}
% \bibliographystyle{plain}

\printbibliography

\end{document}